%% file: main_journal.tex
\documentclass[letterpaper, 10 pt, conference]{ieeeconf}  % Comment this line out if you need a4paper
\IEEEoverridecommandlockouts                              % This command is only needed if 
\usepackage{epsfig}
\usepackage{url}
\usepackage{amsmath}

\usepackage{amsfonts}

\usepackage{amsthm}
\usepackage{amssymb}
\usepackage{graphicx}
\usepackage{color}
\usepackage{balance} % For balanced columns on the last page
\usepackage{algorithm2e}
\usepackage{algorithmic}
\usepackage{bm}

\newtheorem{definition}{Definition}

\let\labelindent\relax

\usepackage{enumitem}
\usepackage{graphicx} % Required for the inclusion of images
\graphicspath{{figure/}}
\usepackage{float}
\usepackage[caption=false,font=footnotesize]{subfig}
\usepackage{soul}
\usepackage{color}

\usepackage{algorithm,algorithmic}
\usepackage{cite}

\allowdisplaybreaks

\usepackage{mathtools}

\newtheoremstyle{bfnote}%
{}{}%
{\itshape}{}%
{\bfseries}{.}%
{ }%
{\thmname{#1}\thmnumber{ #2}\thmnote{ (#3)}}
\theoremstyle{bfnote}
\newtheorem{thm}{Theorem}

\newtheorem{lem}{Lemma}

\newtheorem*{base*}{Base Controller}

\usepackage{tikz}

\DeclareSymbolFont{bbold}{U}{bbold}{m}{n}
\DeclareSymbolFontAlphabet{\mathbbold}{bbold}

\DeclarePairedDelimiterX\Set[2]{\lbrace}{\rbrace}%
{ #1 \,\delimsize| \,\mathopen{} #2 }

\usepackage{bookmark}

\usepackage{amsthm}

\newtheorem{assumption}{Assumption}

\title{\LARGE \bf Data-Driven Successive Linearization for Optimal Voltage Control}

\author{Yiwei Dong, Wenqi Cui, Han Xu, Adam Wierman, Steven Low% <-this % stops a space
\thanks{Yiwei Dong and Wenqi Cui are with the Department of Electrical and Computer Engineering, New York University, NY 11201, USA.  }
\thanks{Han Xu, Adam Wierman, and Steven Low, are with the Department of Computing + Mathematical Sciences, California Institute of Technology, CA 91125, USA.} %
\thanks{The authors are partially supported by the Henry Luce Foundation, the Resnik Sustainability Institute, the PIMCO Foundation.}}

\begin{document}
\maketitle
\thispagestyle{empty}
\pagestyle{empty}

\begin{abstract}

Power distribution systems are increasingly exposed to large voltage fluctuations driven by intermittent renewable generation and time-varying loads (e.g., electric vehicles and storage). To address this challenge, a number of advanced controllers have been proposed for voltage regulation. However, these controllers typically rely on fixed linear approximations of voltage dynamics. As a result, the solutions may become infeasible when applied to the actual voltage behavior governed by nonlinear power flow equations, particularly under heavy power injection from distributed energy resources. This paper proposes a data-driven successive linearization approach for voltage control under nonlinear power flow constraints. By leveraging the fact that the deviation between the nonlinear power flow solution and its linearization is bounded by the distance from the operating point, we perform data-driven linearization around the most recent operating point. Convergence of the proposed method to a neighborhood of KKT points is established by exploiting the convexity of the objective function and structural properties of the nonlinear constraints. Case studies show that the proposed approach achieves fast convergence and adapts quickly to changes in net load.
\end{abstract}


\section{Introduction}
\input{introduction_journal}

\section{Model and Motivations} \label{sec:model}
\input{model_journal}

\section{Case Study} \label{sec:simulation}
\input{simulation_journal}

\section{Conclusion} \label{sec:conclusion}
This paper proposes a data-driven successive linearization approach for voltage control under nonlinear power flow constraints. By conducting local linearization using Jacobian estimation from recent measurements, the method adapts to the optimization landscape of the nonconvex voltage control problem without requiring explicit power network models. Convergence guarantees to a neighborhood of KKT points are established through the design of trust-region mechanism. 
Numerical experiments on the IEEE 33-bus system demonstrate that the proposed approach achieves fast convergence to a low cost solution and adapts quickly to changes in net load. Directions for future work include distributed implementation under limited communication capabilities, verification on distribution systems with three-phase imbalances, and rigorous convergence analysis for time-varying systems subject to measurement noise.

%% the bibliography file.
\bibliographystyle{IEEEtran}
%\bibliography{sample-base}
\bibliography{Reference}
\appendix
\input{appendix}
\end{document}

%% file: introduction_journal.tex
Power distribution systems need to maintain voltage deviations within 5\%-10\% around their nominal value for safe operation of the system~\cite{turitsyn_options_2011}. The intermittency of distributed energy resources (e.g., solar photovoltaic generation), together with sudden variations in demand  (e.g., electric vehicles and storage), can induce large voltage fluctuations. At the same time, the power electronic interface from inverter-based resources provides controllability to regulate voltages through the changes of reactive power injections~\cite{turitsyn_options_2011,yeh2012adaptive,li2014real,bolognani2013distributed,gupta2024optimal}. A large body of existing research has focused on developing voltage control strategies based on the LinDistFlow model, which is a linearized power flow model that characterizes the relation between voltage magnitude and power injections across the distribution network~\cite{farivar_equilibrium_2013, zhu_fast_2016,li_real-time_2014, zhang_local_2013, cui_leveraging_2025}. This offers tractability and computational efficiency for deriving optimal voltage control policies. However, such linear approximations often fail to capture the nonlinear nature of the power flow model, leading to solutions that can become infeasible or ineffective in practice. This issue becomes increasingly common and pronounced, particularly under light net loading conditions (i.e., power demand is not significantly larger than power generation), a scenario that is becoming more prevalent with the growing penetration of distributed energy resources (DERs)~\cite{nazir2021grid, huang2016sufficient}.

Optimal power flow-based formulations have been proposed to solve voltage control with nonlinear system models~\cite{farivar2013branch,zhang2014optimal,low2014convex}. Several convex relaxations have been developed to transform the resulting nonconvex optimization problems into convex ones, such as second-order cone programs or semidefinite programs~\cite{farivar2013branch,zhang2014optimal,low2014convex}. While the action obtained through convex relaxation is exact when the system is radial and heavily loaded, the optimality gap might be large otherwise\cite{low2014convex2,nazir2021grid}. 
In addition, solving the convex relaxation usually requires the exact information of the topologies and parameters of the distribution system, which may not be available in many distribution systems. 
 % In addition, it has been shown that when the system is not heavily loaded, the actions obtained convex relaxation may render high cost 

To overcome this challenge, data-driven approaches have been proposed~\cite{chen_data-driven_2020,chen2021reinforcement,yeh_robust_2022,yuan2023learning}. Model-free reinforcement learning-based methods have been developed in~\cite{feng_bridging_2023,shi2021stability,cui2022decentralized,yuan2023learning} to learn voltage controllers from data collected through interaction with the system. Their  performance guarantees are generally built on linearized system models (i.e., LinDistFlow model). To handle the nonlinearity of the power flow models, a convex neural network is adopted in~\cite{chen_data-driven_2020} to learn the nonlinear mapping between voltage deviations and power injections. The neural network-based model is subsequently embedded into voltage optimization problem to find optimal power injections. For efficient utilization of online data,~\cite{dominguez2023online,bianchi2025stability,ortmann2026subtransmission} propose feedback optimization approaches to update control actions based on sensitivity estimation 
of changes in bus voltage magnitudes with respect to changes
in reactive power injections. Recent studies have demonstrated that feedback optimization can perform reliably in real-world distribution grids, and feedback-based safe gradient flow can additionally provide voltage constraint satisfaction~\cite{ortmann2020experimental,HAUSWIRTH2024100941,10777921}. However, these approaches typically rely on gradient descent-based iterations for solution updates, where the step size cannot be very large due to numerical stability issues and therefore potentially limit their ability to quickly adapt to time-varying conditions during real-time operation. In addition, the distance between the converged solution and the original nonlinear optimization problem is generally not guaranteed.

More generally, similar challenges in handling nonlinear dynamics show up in many other control problems~\cite{petrov2014modeling, tordesillas2019faster, malyuta2022convex,mao2018successive}, including autonomous driving, space vehicle control, motion planning, etc. In particular,~\cite{mao2018successive} proposes to successively convexify the nonlinear control problem by linearizing the dynamics around the iterates of operating points. The convergence to the  Karush–Kuhn–Tucker (KKT) point is established.  Sequential convex programming is also adopted in~\cite{zhang2025towards} to solve nonlinear optimal power flow-based voltage control.  However, this approach depends on the exact model information. This paper seeks to address the following question: Can we solve the optimal control problem through successive linearization in a data-driven manner, while still achieving convergence guarantees?

In this paper, we propose a data-driven successive linearization approach to address voltage control with nonlinear power flow models. By leveraging the fact that a nonlinear system is not far away from its linearization when the operating point is close to the linearization point, we perform data-driven linearization around the most recent operating point through using the information of past trajectory measurements. This procedure yields a sequence of data-driven convex surrogates along the trajectory, each of which is solved to optimality at every iteration. Convergence results to a neighborhood of KKT points are established by bounding the discrepancy between the original nonlinear model and its convex surrogate through trust regions. Case study demonstrates that the proposed method achieves faster convergence and adapts quickly to changes in net load. 

Contributions of this work are summarized as follows:
\begin{enumerate}
    \item We propose a data-driven successive linearization framework for voltage control with nonlinear power flow models, which leverages the most recent operating trajectories to effectively approximate nonlinear system behavior near the current operating point.
    \item The proposed framework establishes a formal convergence guarantee to a neighborhood of  KKT points of optimal voltage control problem with nonlinear power
flow models.
%     bound on the distance to the KKT set of optimal voltage control problem with nonlinear power
% flow models.
    \item Case studies show that the proposed method converges rapidly and adapts quickly to changes in net load. It achieves costs comparable to, or lower than, those obtained using model-based convex relaxations.
    % achieves cost function that is no worse than that obtained by model-based convex relaxations.
\end{enumerate}

\textbf{Notation.}
 Throughout this manuscript, vectors are denoted in lower-case bold, matrices are denoted in upper-case bold, and scalars are unbolded. 
Subscript $(i)$ indicates variables for the node $i$. For variables $a_{(1)}, \cdots, a_{(n)}$,
% , each associated with one of the $n$ nodes, 
$\bm{a}:=(a_{(1)},\cdots, a_{(n)})$ denotes the stacked column vector, and $\bm{a}_k$ denotes its value at iteration $k$. Unless specified otherwise, $\|\cdot\|$ denotes the Euclidean two-norm.

%% file: model_journal.tex
\subsection{Model }

Consider a distribution system with $n$ buses, let $y_{(i)}$ be the voltage at bus $i$ and $\bm{y}:=\left(y_{(1)},\cdots,y_{(n_o)}\right)$ be the $n_o$ observed and synchronously measured voltages across the network. Let $\bm{u}:=\left(u_{(1)},\cdots,u_{(n_c)}\right)$ be the $n_c$ controllable reactive power across the network, and $\mathcal{U}$, $\mathcal{Y}$ represent the convex and compact feasible set for $\bm{u}$ and $\bm{y}$, respectively. Our objective is to control the reactive power injection to reduce voltage deviation from its reference value $\bm{y}^{ref}$, while reducing the incurred control effort. Let $\bm{y}=\bm{g}(\bm{u})$ represents the relation 
between the controllable reactive power injections and observed voltage magnitudes, we formulate the voltage control problem as
% \left( \|\bm{q}_{t}\|_1+\gamma\|\bm{u}_k\|_1\right)
\vspace{-0.2cm}
\begin{equation}\label{eq:Optimization}
\begin{split}
\min_{\bm{u}, \bm{y}}  &\quad  c_v(\bm{y})+c_u(\bm{u})\\
% \label{subeq:Optimization_obj}\\
\mbox{s.t. } 
&\quad \bm{y} = \bm{g}(\bm{u}),\ \  \bm{u}\in\mathcal{U},\ \  \bm{y}\in\mathcal{Y}
% \\
% & \underline{u}_i\leq u_i\leq \bar{u}_i
\end{split}
\end{equation}
where $c_v(\cdot)$ and $c_u(\cdot)$ are convex cost functions on voltage deviations and control effort, respectively. The costs depend on the type of resources and can be either quadratic~\cite{zhao2014design, mallada2017optimal} or non-quadratic~\cite{shi2017using, vaccaro2011decentralized, jafari2018optimal}. We assume the feasibility set of the problem is nonempty.

The mapping from reactive power injections to voltage magnitudes, $\bm{y}=\bm{g}(\bm{u})$ , is defined implicitly by the nonlinear DistFlow equations~\cite{baran2002optimal,chiang2002existence}. In particular, for a given $\bm{u}$, the voltage is obtained by solving the following equations:
\vspace{-0.2cm}
\begin{subequations}\label{eq:voltage}
\begin{align}
& y_{(i)}^{2} = y_{(j)}^{2} + 2 r_{ij} P_{ij} + 2 x_{ij} Q_{ij} - |z_{ij}|^{2} l_{ij},\\
& \sum_{h:\ j \rightarrow h} P_{jh} = P_{ij} - r_{ij} l_{ij} + p_{(j)},\\
& \sum_{h:\ j \rightarrow h} Q_{jh} = Q_{ij} - x_{ij} l_{ij} + u_{(j)},\\
&  l_{ij} = \left(P_{ij}^{2} + Q_{ij}^{2}\right)/y_{(i)}^{2}\label{subeq:loss}
\end{align}
\end{subequations}
where $i,j\in\{1,2,\dots,n\}$ and $p_{(j)}$ denotes the active power injection at node $j$.
For the line connecting nodes $i$ and $j$, $P_{ij}$ and $Q_{ij}$ denote the active and reactive power
flows, $r_{ij}$ and $x_{ij}$ are the line resistance and reactance, and $l_{ij}$ is the squared current
magnitude on the line. The complex impedance is $z_{ij}=r_{ij}+\mathrm{j}x_{ij}$, and the notation
$h:\,j\rightarrow h$ denotes the set of children of node $j$. We assume that the operating points
are not near the voltage-collapse boundary; namely, the Jacobian of the DistFlow equation with respect to state variables is nonsingular and not close to singular. This is a common assumption for voltage control, which ensures that the solution of~\eqref{eq:voltage} exists~\cite{bolognani2015existence}. Together with the Implicit Function
Theorem~\cite{jindal2024implicit_domain}, it implies that, for fixed exogenous active power
injections and other uncontrollable injections, the DistFlow equations induce a locally unique and
continuously differentiable mapping $\bm g:\bm u\mapsto\bm y$. Note that it is generally difficult to explicitly characterize the function $\bm{y}=\bm{g}(\bm{u})$ from~\eqref{eq:voltage}, and one of the benefits brought by the proposed method is that we do not require the analytical form of $\bm{g}$ as prior knowledge.

A series of works has been established in solving~\eqref{eq:Optimization}~\eqref{eq:voltage} through convex relaxation of the DistFlow~\cite{farivar2013branch,zhang2014optimal,low2014convex}. For example, converting the line loss function in~\eqref{subeq:loss} to inequality constraints can convert the problem into a second-order cone program. The exactness of convex relaxation has also been established~\cite{farivar2013branch}. However, these approaches rely on exact model information, which may not be available for many distribution networks. 

To eliminate the reliance on exact model information, we consider data-driven methods that update control actions based on system measurements. Existing methods typically use data measurement to approximate the gradient $\nabla_{ \bm{u}_k}c(\bm{u}_k)$  and update control actions iteratively through $\bm{u}_{k+1}=\text{Proj}_{\mathcal{U}}(\bm{u}_{k}-\eta \nabla_{ \bm{u}_k}c(\bm{u}_k))$ with $\text{Proj}_{\mathcal{U}}$ being the projection operator onto ${\mathcal{U}}$ and $\eta$ being the learning rate \cite{dominguez2023online,bianchi2025stability,ortmann2026subtransmission,11066276,kinga2024new}. For example, feedback optimization computes the gradient by $\nabla_{ \bm{u}_k}c(\bm{u}_k)=
\nabla_{ \bm{u}_k}\bm{g}(\bm{u}_k)\nabla_{ \bm{y}_k} c_v(\bm{y}_k)+\nabla_{ \bm{u}_k}c_u(\bm{u}_k)$, where the Jacobian matrix $\nabla_{ \bm{u}_k}\bm{g}(\bm{u}_k)$ is calculated from the online input-output data~\cite{dominguez2023online,bianchi2025stability,ortmann2026subtransmission,10777921}. However, gradient descent-based updates, if not specifically tailored to the specific operating point, may require a long horizon to converge. In addition, the quality of the converged solution is generally not guaranteed.

\subsection{Successive Linearization}
To avoid the slow convergence of gradient descent in data-driven approaches, we propose to leverage the structure of the problem  that the only non-convexity in problem~\eqref{eq:Optimization} comes from the nonlinear DistFlow in~\eqref{eq:voltage}. 
Linearizing the DistFlow will convert~\eqref{eq:Optimization} to a convex problem, which allows us to solve it to optimality without being restricted to gradient-descent-type updates. Motivated by this, we solve the problem~\eqref{eq:Optimization} by successively linearizing the dynamical system at the operating point $\bm{u}_{k}$ for iterations $k=1,2,\cdots, N$. Specifically, we initialize $\bm{u}_1$ as the current operating point.  For each $k=1,\cdots, N-1$,  $\bm{u}_{k+1}$ is obtained by solving the optimization problem 
\vspace{-0.2cm}
\begin{subequations}\label{eq:Optimization-lin0}
\begin{align}
\min_{\bm{u},\bm{y}}  &\quad  c_v(\bm{y})+c_u(\bm{u})\label{subeq:Optimization_obj}\\
\mbox{s.t. } & \quad\bm{y}=\bm{y}_{k}+\nabla_{ \bm{u}_k}\bm{g}(\bm{u}_k)\left(\bm{u}-\bm{u}_{k}\right),\label{subeq:model_lin}\\
&\quad\bm{u}\in\mathcal{U},\ \  \bm{y}\in\mathcal{Y}
\end{align}
\end{subequations}
where $\nabla_{ \bm{u}_k}\bm{g}(\bm{u}_k)$ is the Jacobian of the function $\bm{g}(\bm{u})$ at the operating point $\bm{u}_k$.

For trajectories that remain sufficiently close to the operating point $\bm{u}_k$, it is reasonable to regard them as being generated by the local linearization in~\eqref{subeq:model_lin}, which enables the use of trajectory data to construct a data-driven linear approximation of the underlying nonlinear system.
However, two fundamental challenges arise. First, data-driven approximations inevitably introduce modeling errors in the resulting linearized dynamics. Second, even when locally accurate, the resulting linearized model may deviate substantially from the original nonlinear DistFlow equations, potentially leading to large optimality gaps when solutions derived from the linear model are applied to the nonlinear problem. These challenges make it nontrivial to ensure convergence of data-driven linearized control schemes.

 To address these issues, we develop a data-driven successive linearization framework equipped with a trust region mechanism. This approach enables efficient use of the most recent trajectory data to iteratively refine the local linearized model while ensuring that the induced modeling errors remain bounded. In the following sections, we present the design of the proposed data-driven successive linearization method for the optimal voltage control problem and demonstrate how convergence guarantees are achieved through the trust region mechanism.

\section{Data-driven Successive Linearization}
In this section, we first introduce the overall framework for data-driven successive linearization with a trust region mechanism. We then present the detailed computational procedures and the corresponding algorithmic design.  
\subsection{Successive Linearization with Trust Region Mechanism}
To explicitly compute the Jacobian $\nabla_{ \bm{u}_k}\bm{g}(\bm{u}_k)$ with respect to each operating point, we will need the information of the model $\bm{g}(\bm{u})$ and its parameters, which is often difficult to obtain.
% It is often difficult to be obtained for systems without exact model information or parameters. 
Therefore, to avoid relying on such model information, we instead resort to a data-driven approach to estimate $\nabla_{ \bm{u}_k}\bm{g}(\bm{u}_k)$ through the past trajectory $\left\{\left(\bm{y}_{k-\tau}, \bm{u}_{k-\tau}\right), \cdots, \left(\bm{y}_k, \bm{u}_k\right)\right\}$.

We consider the estimation of $\nabla_{ \bm{u}_k}\bm{g}(\bm{u}_k)$ as finding the sensitivity matrix ${\mathbf{S}}_{k}$ that maps perturbations of the input $\bm{u}$ to the changes of output $\bm{y}$:
\vspace{-0.2cm}
$$
\Delta \bm{y}_{k} \approx {\mathbf{S}}_{k} \Delta \bm{u}_{k},
$$
where the details of computing the sensitivity matrix will be given in the next subsection.

Combined with a trust region mechanism to bound the changes of $\bm{u}$, the data-driven approach to solving~\eqref{eq:Optimization} is thus formulated as
\vspace{-0.2cm}
\begin{subequations}\label{eq:Optimization-linData}
\begin{align}
\min_{\bm{u}\in\mathcal{U},\bm{y}\in\mathcal{Y}}  \quad&  c_v(\bm{y})+c_u(\bm{u})\label{subeq:Optimization_obj}\\
\mbox{s.t. } \quad&  \bm{y}=\bm{y}_{k}+\widetilde{\mathbf S}_{k} \left(\bm{u}-\bm{u}_{k}\right),\label{subeq:datadriven_lin}\\
&\|\bm{u}-\bm{u}_{k}\|\leq r_{k},\label{subeq:reust_region}
\end{align}
\end{subequations}
where $\widetilde{\mathbf S}_k$ is the estimation of the Jacobian $\nabla_{ \bm{u}_k}\bm{g}(\bm{u}_k)$, $r_{k}$ is the radius of the trust region for updating $\bm{u}$ at iteration $k$. The equation~\eqref{subeq:datadriven_lin} represents the data-driven linearized power flow constraint around the operating point $\bm{u}_k$. The trust region design is motivated by~\cite{mao2018successive} and plays a central role in guaranteeing convergence to the solution of problem~\eqref{eq:Optimization} (detailed theorem and proof will be given in Section~\ref{sec:convergence}).
 Iterating this process gives the sequence of solutions $\left\{\left(\bm{y}_1, \bm{u}_1\right), \cdots, \left(\bm{y}_N, \bm{u}_N\right)\right\}$. 

Although the trust region is a hyperparameter to be tuned, its scheme to update control action differs fundamentally from gradient descent-based approaches. In particular, the trust-region subproblem~\eqref{eq:Optimization-linData} is solved to optimality at every step, enabling it to more effectively exploit the local optimization landscape. Benefiting from this, the proposed approach is less sensitive to parameter tuning and converges faster. In contrast, gradient-based methods usually require a conservative stepsize to avoid oscillations or instability, which can significantly slow convergence. We demonstrate these differences in case studies.

\subsection{Computation of the Sensitivity Matrix}\label{sec:sens}

We estimate the Jacobian $\nabla_{\bm u_k}\bm g(\bm u_k)$ by solving a weighted least-squares (WLS) problem~\cite{dominguez2023online}:
\vspace{-0.2cm}
\begin{equation}
    \widetilde{\mathbf S}_{k}
    = \underset{\mathbf S}{\operatorname{argmin}}\sum_{l=k-\tau+1}^{k}
    \omega^{\,k-l}\big\|\Delta\bm y_l-\mathbf S\,\Delta\bm u_l\big\|^{2},
\label{eqsenloss}
\end{equation}
where 
$\Delta\bm u_l:=\bm u_l-\bm u_{l-1}$ and
$\Delta\bm y_l:=\bm y_l-\bm y_{l-1}$
denote the first-order differences of the reactive power injections and bus
voltages, respectively. The parameter $\tau$ is the length of the window used
for sensitivity estimation, and $\omega\in(0,1)$ is the forgetting factor that
assigns less weight to older measurements.

To express the WLS solution to~\eqref{eqsenloss} in matrix form, we stack
$\Delta\bm u_l$ and $\Delta\bm y_l$ as rows and define
\begin{equation}\label{eq:U_Y_def}
\bm U:=
\begin{bmatrix}
\Delta\bm u_{k-\tau+1}^\top\\[-.25em]
\vdots\\[-.25em]
\Delta\bm u_{k}^\top
\end{bmatrix}
\in\mathbb R^{\tau\times n_c}, \ 
\bm Y:=
\begin{bmatrix}
\Delta\bm y_{k-\tau+1}^\top\\[-.25em]
\vdots\\[-.25em]
\Delta\bm y_{k}^\top
\end{bmatrix}
\in\mathbb R^{\tau\times n_o}.
\end{equation}
The forgetting-factor weights can then be collected into a diagonal matrix
$\bm W\in\mathbb R^{\tau\times\tau}$ with
\begin{equation}\label{eq:Wk_diag}
\bm W
:=
\mathrm{diag}\big(
\omega^{\tau-1},\omega^{\tau-2},\dots,\omega,1
\big),
\end{equation}
so that the cost in~\eqref{eqsenloss} can be written as
\begin{equation}
    \ell(\mathbf S)
    =
    \big\|\bm W^{1/2}(\bm Y-\bm U\mathbf S^\top)\big\|_F^2,
\end{equation}
where $\|\cdot\|_F$ denotes the Frobenius norm.

The length $\tau$ of the past trajectory for computing the sensitivity estimation is determined based on sufficient excitation condition of matrix $\bm U$ defined as follows:

\begin{definition}[Sufficient excitation]\label{defsuff}
  The past trajectory $\left\{\Delta\bm u_l, \Delta\bm y_l \right\}$ for
  $l=k-\tau+1,\dots,k$ is said to be sufficiently excited if the matrix
  $\bm U$ in~\eqref{eq:U_Y_def} is full column rank.
\end{definition}

When the sufficient excitation condition holds, the WLS estimator admits the
closed-form solution
\begin{equation}\label{eq:S_explicit_WLS}
    \widetilde{\mathbf S}_k^\top
    =
    \big(\bm U^\top \bm W \bm U\big)^{-1}
    \bm U^\top \bm W \bm Y.
\end{equation}

Ideally, perturbations of $\bm{u}$ should be conducted around the operating point of interest. This would require the system to be `fixed' around the point and add small perturbations until the sufficient excitation of $\bm{U}$ is satisfied.  
% without conducting actions that change the operating point. 
However, such a procedure would substantially increase the amount of data required and slow convergence. Therefore, we hope that the system can perform optimization at each action step, rather than keeping the system fixed while perturbing around the current operating point. Notably, we find that the difference between the estimated sensitivity and the true Jacobian can be bounded by the trust region $r_k$. Properly setting the trust region allows us to use the past trajectory efficiently and still achieve convergence guarantees. Details will follow in Section~\ref{sec:convergence}.

\subsection{Algorithm}
The data-driven successive linearization method for solving \eqref{eq:Optimization} proceeds as follows. At iteration $k$, we estimate the local Jacobian matrix via weighted least squares~\eqref{eqsenloss} using recent input–output increments. This gives us the data-driven linearized model $\bm{y}=\bm{y}_k+{\widetilde {\mathbf{S}}}_k(\bm{u}-\bm{u}_k)$. Solving the trust-region  subproblem~\eqref{eq:Optimization-linData} to optimality then produces the next iterate $\bm{u}_{k+1}$. After applying $\bm{u}_{k+1}$ and measuring $\bm{y}_{k+1}$, we refresh the data collection buffer and repeat the above process. The complete procedure of the data-driven successive linearization method is summarized in Algorithm~\ref{alg:ddsl}.
\begin{algorithm}[t]
\caption{Data-Driven Successive Linearization}
\label{alg:ddsl}
\begin{algorithmic}[1]
\STATE \textbf{Input:} Constraint sets $\mathcal{U},\mathcal{Y}$; costs $c_v(\cdot), c_u(\cdot)$; previous data window of length $\tau+1$; forgetting factor $\omega\!\in\!(0,1)$; initial $(\bm{u}_1,\bm{y}_1)$; trust-region radius $r_k$ 
\FOR{$k=1,2,\dots,N-1$}
    \STATE collect $\{\Delta\bm{u}_\ell,\Delta\bm{y}_\ell\}_{\ell=k-\tau+1}^{k}$ with
    $\Delta\bm{u}_\ell=\bm{u}_\ell-\bm{u}_{\ell-1}$ and
    $\Delta\bm{y}_\ell=\bm{y}_\ell-\bm{y}_{\ell-1}$
   
    \STATE solve
    \begin{equation*}
    \widetilde{\mathbf S}_{k} = \arg\min_{\mathbf S}\; \sum_{\ell=k-\tau+1}^{k}\omega^{\,k-\ell}
    \left\|\Delta{\bm{y}}_\ell - {\mathbf S}\,\Delta{\bm{u}}_\ell\right\|^2
\end{equation*}
    where $(\Delta{\bm{u}}_\ell,\Delta{\bm{y}}_\ell)$ are measured increments
    \STATE compute $\bm{u}_{k+1}$ by solving the linearized convex trust-region subproblem
    \begin{equation*}
    \begin{aligned}
    &\min_{\bm{u}\in\mathcal{U},\bm{y}\in\mathcal{Y}}\;\;  c_v(\bm{y})+c_u(\bm{u})\\
    \text{s.t.}\;\; &\bm{y}=\bm{y}_k + \widetilde{\mathbf S}_{k}(\bm{u}-\bm{u}_k),\ \   \|\bm{u}-\bm{u}_k\| \le r_k
    \end{aligned}
    \end{equation*} 
    \STATE implement $\bm{u}_{k+1}$ and measure $\bm{y}_{k+1}$
\ENDFOR
\STATE \textbf{Output:} sequences $\{\bm{u}_k\}_{k=1}^{N}$ and $\{\bm{y}_k\}_{k=1}^{N}$
\end{algorithmic}
\end{algorithm}

\section{Convergence Analysis}\label{sec:convergence}

This section establishes convergence of the data-driven successive linearization method. We first introduce notations and assumptions, and then present the main theorem and proofs for the data-driven successive linearization.

\subsection{Main Theorem}

Throughout the convergence analysis, the exogenous active power injections and uncontrollable reactive power injections are assumed to remain fixed across the algorithm iterations. Specifically, the following assumption holds: 
\begin{assumption}[Fixed exogenous condition]
\label{assump:fixed_exogenous}
The DistFlow equations induce a fixed mapping $\bm y=\bm g(\bm u)$, and the feasible sets $\mathcal U$, $\mathcal Y$ remain unchanged across iterations.
\end{assumption}

For notational convenience, let $\bm z:=(\bm y,\bm u)$ and define the constraint residual $\bm{e}(\bm z):=\bm y-\bm g(\bm u)$. In this way, the equations in problem~\ref{eq:Optimization} are converted to the equality constraints $\bm{e}(\bm z)=\bm 0$, with $\bm e(\cdot)$ nonconvex and continuously differentiable. 
Furthermore, we denote $\{\bm{z}:\bm{h}(\bm z)\le 0\}$ as the compact constraint set for $\bm z$, where each $h_j$ is convex and continuously differentiable. The requirement for cost functions is specified in the following assumption.
% \begin{assumption}[Normal operating points]
% \label{ass:away_from_collapse}
% All operating points stay away from voltage-collapse boundaries.
% \end{assumption}

% Voltage-collapse boundaries in Assumption~\ref{ass:away_from_collapse} refer to points where the Jacobian of the DistFlow equations with respect to the voltage variables becomes (near) singular. At such points, the voltage solution may become extremely sensitive to perturbations in $\bm u$, meaning that small changes in the reactive power injection can lead to disproportionately large changes in voltages. At normal operating points, the Implicit Function Theorem~\cite{jindal2024implicit_domain} guarantees that the DistFlow equations induce a locally unique solution mapping $\bm g:\bm u\mapsto\bm y$ that is continuously differentiable. Futhermore, under Assumption~\ref{ass:away_from_collapse}, the sensitivity $\nabla_{\bm u} \bm g(\bm u)$ exists and is locally Lipschitz in a neighborhood of $\bm u$. We assume the following regularity for the cost.

\begin{assumption}[Cost regularity]\label{assump:cost}
 The cost functions $c_v$ and $c_u$ are convex and continuously differentiable. The composite objective with respect to $\bm z$, $c(\bm z) := c_v(\bm y)+c_u(\bm u)$, is also convex and continuously differentiable.
\end{assumption}
We additionally assume that the Linear Independence Constraint Qualification (LICQ) holds, which is a standard regularity condition in nonlinear optimization and is commonly imposed to establish convergence and stationarity results for nonconvex constrained problems~\cite{mao2018successive,verschueren2016exploiting,doi:10.1137/110844349}.

\begin{definition}[LICQ]
\label{def:LICQ}
Let $\mathcal J_{\mathrm{eq}}$ index the equality constraints in $\bm{e}(\bm z)=\bm 0$ and
$\mathcal J_{\mathrm{ineq}}$ index the inequalities $\{h_j(\bm z)\le 0\}$. Define the active set
$\mathcal J_{\mathrm{ac}}(\bm z):=\{\,j\in\mathcal J_{\mathrm{ineq}}:\ h_j(\bm z)=0\,\}$.
LICQ holds at $\bm z$ if the set of gradients 
$$
\{\nabla e_i(\bm z):\, i\in \mathcal J_{\mathrm{eq}}\}\ \cup\
\{\nabla h_j(\bm z):\, j\in \mathcal J_{\mathrm{ac}}(\bm z)\}
$$
is linearly independent.
\end{definition}

If $\bar{\bm z}$ is a local optimum of problem~\eqref{eq:Optimization}
and LICQ holds at $\bar{\bm z}$, then $\bar{\bm z}$ satisfies the KKT
conditions.

\begin{definition}[KKT conditions for \eqref{eq:Optimization}]\label{def:KKT}
A feasible point $\bar{\bm z}$ is a KKT point of \eqref{eq:Optimization} if there exist multipliers $\{\nu_i\}_{i\in\mathcal J_{\mathrm{eq}}}$ and $\{\mu_j\}_{j\in\mathcal J_{\mathrm{ac}}(\bar{\bm z})}$ with $\mu_j\ge 0$ such that
$$
\nabla c(\bar{\bm z})
\;+\;\sum_{i\in\mathcal J_{\mathrm{eq}}}\nu_i\,\nabla e_i(\bar{\bm z})
\;+\;\sum_{j\in\mathcal J_{\mathrm{ac}}(\bar{\bm z})}\mu_j\,\nabla h_j(\bar{\bm z})
\;=\;\bm 0 .
$$
\end{definition}

Definition~\ref{def:KKT} characterizes first-order optimality for the original constrained problem~\ref{eq:Optimization}. For the convergence analysis of the data-driven successive linearization, it is convenient to work with an exact-penalty reformulation that aggregates the constraints into the objective, and then find the corresponding optimality conditions. We therefore introduce the penalty problem and the associated stationary set.

For the exact penalty objective, we minimize the cost:
\begin{equation}\label{eq:penalty_cost}
 J(\bm z)
  \;:=\;
  c(\bm z)\;+\;\bm\lambda^\top \!\big|\bm{e}(\bm z)\big|
  \;+\;\bm\eta^\top \!\big[\bm h(\bm z)\big]_+ ,
\end{equation}
with penalty vectors $\bm\lambda\ge \bm 0,\ \bm\eta\ge \bm 0$, and $[\bm h(\bm z)]_+:=\max\{\bm h(\bm z),\bm 0\}$ is applied componentwise. 
The stationary set for \(\min_{\bm z} J(\bm z)\) is
\(\mathcal T:=\{\bm z:\;\bm 0\in \partial J(\bm z)\}\), where \(\partial J\) denotes the generalized differential.
By Theorem~3.9 of~\cite{mao2018successive}, for sufficiently large penalties,
any stationary point \(\bar{\bm z}\in \mathcal T\) that is feasible for \eqref{eq:Optimization}
is a KKT point of the original problem.
Thus it suffices to show that the data-driven successive linearization iterates approach \(\mathcal T\) for \eqref{eq:penalty_cost}. Denote the distance between a point $\bm z$ and a set $\mathcal D$ as $\mathrm{dist}(\bm z,\mathcal D)
    :=
    \inf_{\bm x\in\mathcal D}\|\bm z-\bm x\|$, we have the following theorem.

\begin{thm}\label{thmmain}
Let Assumptions~\ref{assump:fixed_exogenous}--\ref{assump:cost}  hold.
Then the data-driven successive linearization algorithm always has limit points.
Let $\bar{\bm z}$ be any limit point and $\{\bm z_{k_i}\}$ be a subsequence converging to it, then $\bar{\bm z}$ satisfies the stationarity bound 
$\mathrm{dist}\big(\bm 0,\partial J(\bar{\bm z})\big)
\ \le\ \|\bm\lambda\|\,\liminf_{i\to\infty}\|\widetilde{\mathbf S}_{k_i}-\mathbf S_{k_i}\|$ for the penalty problem~\eqref{eq:penalty_cost}. 
 If the Jacobian estimation error vanishes,
then any limit point $\bar{\bm z}$ is a stationary point of~\eqref{eq:penalty_cost}.
Furthermore, if $\bar{\bm z}$ is feasible for the original problem~\eqref{eq:Optimization} and LICQ holds at $\bar{\bm z}$ as in Definition~\ref{def:LICQ},
then $\bar{\bm z}$ is a KKT point of~\eqref{eq:Optimization}.
\end{thm}

We establish the theorem as follows. First, we show that the distance between the data-driven linearized penalty objective and the accurate penalty objective can be bounded by trust region. 
% data-driven linearization is first-order accurate on each trust region. 
This further yields a uniform estimate linking the decrease of the data-driven linearized objective to the true decrease of the original problem in a neighborhood of any nonstationary point at each step. Combining this relation with optimality conditions of the linearized penalty subproblem, we show that the algorithm admits limit points that lie within a neighborhood of stationary points. When the limit point is stationary for the penalized problem and also satisfies the original constraints, exact-penalty theory~\cite{mao2018successive} together with the LICQ assumption implies that it satisfies the KKT conditions of the original problem. The next subsection presents a detailed proof following this outline.

\subsection{Proof of Theorem~\ref{thmmain}}
To relate the cost arising from successive linearization to the cost under the nonlinear model $J(\bm z)$, we introduce the following penalized cost functions.
For iteration step $k$, we have the decision variables $\bm z_k=(\bm y_k,\bm u_k)$, and recall that the original problem enforces the equality constraint
$\bm e(\bm z_k) = \bm y_k - \bm g(\bm u_k) = \bm 0$.
For notational convenience, we denote $\mathbf S_k:=\nabla_{ \bm{u}_k}\bm{g}(\bm{u}_k)$ throughout the remainder of the paper. The proof follows a similar high-level structure to that in~\cite{mao2018successive}, but differs in an essential way due to the Jacobian estimation error introduced by the data-driven procedure. 
% Let $\mathbf S_k:=\nabla\bm g(\bm u_k)$ denote the JacobianJacobian and
% $\widetilde{\mathbf S}_k$ its data-driven estimate from Section~\ref{sec:sens}.

Define the step
\begin{equation}\label{eq:d_def}
\bm d:=(\bm d_y,\bm d_u):=(\bm y-\bm y_k,\ \bm u-\bm u_k).
\end{equation}

The penalized cost associated with the model linearized at $\bm z_k$ is
\begin{equation}\label{eq:Lk_true}
  \begin{aligned}
  &L_k(\bm d)
  \; :=\;
  c_v\!\big(\bm y_k+\bm d_y-\bm y^{\rm ref}\big)
  \;+\; c_u(\bm u_k+\bm d_u)\\[-.25em]
  &\;+\; \bm\lambda^\top\!\big|\bm d_y-\mathbf S_k\,\bm d_u\big|
  \;+\; \bm\eta^\top \!\big[\bm h(\,\bm y_k+\bm d_y,\ \bm u_k+\bm d_u\,)\big]_+ .
  \end{aligned}
\end{equation}

The data–driven counterpart replaces $\mathbf S_k$ by the data-driven estimate $\widetilde{\mathbf S}_k$:
\vspace{-0.2cm}
\begin{equation}\label{eq:Lk_est}
  \begin{aligned}
  &\widetilde L_k(\bm d)
  \;:=\;
  c_v\!\big(\bm y_k+\bm d_y-\bm y^{\rm ref}\big)
  \;+\; c_u(\bm u_k+\bm d_u)\\[-.25em]
  &\;+\; \bm\lambda^\top\!\big|\bm d_y-\widetilde{\mathbf S}_k\,\bm d_u\big|
  \;+\; \bm\eta^\top \!\big[\, \bm h(\,\bm y_k+\bm d_y,\ \bm u_k+\bm d_u\,)\big]_+ .
  \end{aligned}
\end{equation}
Note that $L_k(\bm 0)=\widetilde L_k(\bm 0)=J(\bm z_k)$. The map
$\bm d\mapsto \bm d_y-\widetilde{\mathbf S}_k\bm d_u$ is affine, and $\bm h$ is convex;
since compositions with $|\cdot|$ and $[\cdot]_+$ preserve convexity, $\widetilde L_k$ is a convex surrogate of $J$ in a neighborhood of $\bm z_k$. For the convergence proof, we need the solution $\bm d_k^\star$ of the following convex trust-region subproblem at each iteration $k$:
\vspace{-0.2cm}
\begin{equation}\label{eq:TR_subprob}
\begin{aligned}
\min_{\bm d}\quad & \widetilde L_k(\bm d)\\
\text{s.t.}\quad & \|\bm d\| \le r_k ,
\end{aligned}
\end{equation}
with the trust-region radius $r_k>0$.

To analyze the convergence of the data-driven successive linearization method, we need to control the error incurred by replacing the true sensitivity with its data–driven estimate over the trust region. The next lemma provides a bound on this discrepancy.

\begin{lem}\label{lem:est_ptwise}
Suppose the sensitivity estimation window $\{k-\tau,\dots,k\}$ satisfies $\max_{\ell\in\{k-\tau,\dots,k\}}\|\bm z_\ell-\bm z_k\|\ \le\ \alpha\,r_k$ for some $\alpha>0$ and let
$\widetilde{\mathbf S}_k$ be the weighted least–squares estimate for~\eqref{eqsenloss}. Then, for every $\bm z$ with $\|\bm z-\bm z_k\|\le r_k$, there exists a constant $C_S>0$ , such that
\begin{equation}\label{eq:S_error_bound}
\|\widetilde{\mathbf S}_k-\mathbf S_k\|
\;\le\; C_S r_k.
\end{equation}
Consequently, for every $\bm u$ with $\|\bm u-\bm u_k\|\le r_k$,
\begin{equation}\label{eq:key_ptwise_goal}
\big\|(\widetilde{\mathbf S}_k-\mathbf S_k)(\bm u-\bm u_k)\big\|
\;\le\; C_S r_k^2.
\end{equation}
\end{lem}

We prove the lemma by exploiting a closed-form representation of the estimated sensitivity matrix
$\widetilde{\mathbf S}_k$ constructed from the past trajectories of the control inputs $\bm u$ and voltage states $\bm y$. This enables a precise comparison between $\widetilde{\mathbf S}_k$ and the Jacobian $\mathbf S_k$
by leveraging local Lipschitz continuity and the trust-region bound on recent iterates. The detailed proof is given in Appendix~\ref{app:lem_est_ptwise}.

Lemma~\ref{lem:est_ptwise} gives a quadratic bound on the effect of the Jacobian estimation error over the trust region. The next two lemmas compare the nonlinear penalized objective with its linearized surrogates in this trust region.

\begin{lem}\label{lem:taylor}
For every $\bm{d}$ with $\|\bm d_u\|\le r_k$, there exists a constant $C_J>0$ such that
\begin{equation}\label{eq:J_vs_L}
\big|J(\bm{z}_k+\bm{d})-L_k(\bm{d})\big|
\;\le\; C_J r_k^2.
\end{equation}
\end{lem}

\begin{proof}
The functions $c_v$, $c_u$, and the convex-inequality penalty $\bm\eta^\top[\bm h(\bm y_k+\bm d_y,\bm u_k+\bm d_u)]_+$ appear identically in both $J(\bm z_k+\bm d)$ and $L_k(\bm d)$ and hence cancel in their difference. By a second--order Taylor expansion of $\bm g$ at $\bm u_k$, for any step $\bm d_u$ we can write $\bm g(\bm u_k+\bm d_u)
= \bm g(\bm u_k) + \mathbf S_k \bm d_u + \bm \gamma_k(\bm d_u)$,
where the higher-order residual satisfies $\|\bm \gamma_k(\bm d_u)\| \le \beta\,\|\bm d_u\|^2$ for all $\bm d_u$ with some constant $\beta>0$.
Using the triangle inequality componentwise,
\begin{equation}\label{eq:taylor_penalty}
\begin{aligned}
&\big|J(\bm{z}_k+\bm{d})-L_k(\bm{d})\big|\\
&=
\left|
\bm{\lambda}^\top
\!\Big(
\big|\bm{d}_y-\mathbf{S}_k\bm{d}_u-\bm{\gamma}_k(\bm{d}_u)\big|
-
\big|\bm{d}_y-\mathbf{S}_k\bm{d}_u\big|
\Big)
\right|\\
&\le \|\bm{\lambda}\|\,\|\bm{\gamma}_k(\bm{d}_u)\|
\;\le\; \beta\|\bm{\lambda}\|\,\|\bm{d}_u\|^2
\;\le\; \beta\|\bm{\lambda}\|\,r_k^2,
\end{aligned}
\end{equation}
which proves~\eqref{eq:J_vs_L} with $C_J:=\beta\|\bm\lambda\|$.
\end{proof}

\begin{lem}\label{lem:Ltilde_minus_L}
Suppose the assumptions in Lemma~\ref{lem:est_ptwise} hold.
Then the difference between the data-driven surrogate $\widetilde L_k$
and the exact linearized model $L_k$ satisfies
\begin{equation}\label{eq:Ltilde_vs_L}
\big|\widetilde{L}_k(\bm{d})-L_k(\bm{d})\big|\,\le\,C_{\widetilde L}r_k^2\quad\text{for all }\|\bm d_u\|\le r_k,
\end{equation}
In addition, the following relation holds for $J$ and $\widetilde L_k$:
\begin{equation}\label{eq:J_vs_Ltilde}
\big|J(\bm{z}_k+\bm{d})-\widetilde{L}_k(\bm{d})\big|
\;\le\;C_Mr_k^2\quad\text{for all }\|\bm d_u\|\le r_k.
\end{equation}
\end{lem}

\begin{proof}
By the definition of $L_k$ and $\widetilde{L}_k$,
\[
\widetilde{L}_k(\bm{d})-L_k(\bm{d})
=
\bm{\lambda}^\top
\!\left(
\big|\bm{d}_y-\widetilde{\mathbf{S}}_k\bm{d}_u\big|
-
\big|\bm{d}_y-\mathbf{S}_k\bm{d}_u\big|
\right).
\]
Using the triangle inequality componentwise and then Cauchy--Schwarz inequality,
\begin{equation}\label{eq:abs_cauchy_clean}
\begin{aligned}
    \big|\widetilde{L}_k(\bm{d})-L_k(\bm{d})\big|
\;&\le\;
\bm{\lambda}^\top \big|(\widetilde{\mathbf{S}}_k-\mathbf{S}_k)\bm{d}_u\big|
\; \\ &\le\;
\|\bm{\lambda}\|\,\big\|(\widetilde{\mathbf{S}}_k-\mathbf{S}_k)\bm{d}_u\big\|.
\end{aligned}
\end{equation}

Since $\|\bm d_u\|\le r_k$, Lemma~\ref{lem:est_ptwise} yields
$\big\|(\widetilde{\mathbf{S}}_k-\mathbf{S}_k)\bm{d}_u\big\| \le C_Sr_k^2$, and hence
\eqref{eq:Ltilde_vs_L} follows from \eqref{eq:abs_cauchy_clean} with $C_{\widetilde L}:=\|\bm\lambda\|C_S$. For \eqref{eq:J_vs_Ltilde}, write
\begin{equation}
    \begin{aligned}
        &J(\bm{z}_k+\bm{d})-\widetilde L_k(\bm d)\\
&=
\big(J(\bm{z}_k+\bm{d})-L_k(\bm d)\big)
+
\big(L_k(\bm d)-\widetilde L_k(\bm d)\big).
    \end{aligned}
\end{equation}
Let $C_M:=C_J+C_{\widetilde L}$, the result~\eqref{eq:J_vs_Ltilde} follows from \eqref{eq:J_vs_L}, \eqref{eq:Ltilde_vs_L}, and the triangle inequality.
\end{proof}

The above lemma shows that the penalized objective admits a  local approximation by the data-driven linearized model. Since the proposed data-driven successive linearization solves the convex trust-region subproblem to optimality, we can relate the cost decrease achieved on the nonlinear penalized objective to the decrease suggested by the linearized model within the same step. For any point $\bm z$ and step $\bm d$, define the decreases in the nonlinear penalized objective and its data-driven linearized counterpart as

\begin{equation}\label{eq:Delta_defs}
\begin{aligned}
\Delta J(\bm z,\bm d)
&:= J(\bm z)-J(\bm z+\bm d),\\
\Delta \widetilde L(\bm z,\bm d)
&:= J(\bm z)-\widetilde L(\bm d).
\end{aligned}
\end{equation}
The ratio
\begin{equation}\label{eq:rho_def}
\rho(\bm z,r)
:=
\frac{\Delta J(\bm z,\bm d)}{\Delta \widetilde L(\bm z,\bm d)}\, 
\end{equation}
quantifies the agreement between the decrease of the nonlinear penalized objective and the decrease measured by the data-driven linearized model along the step $\bm d$. 
The next lemmas relate $\Delta J(\bm z,\bm d)$ and $\Delta \widetilde L(\bm z,\bm d)$ obtained from the data-driven linearized trust-region subproblem.

\begin{lem}
\label{lem:pred_nonneg}
For any iteration $k$ and radius $r_k>0$, let $\bm d_k^\star$ be a minimizer of the
trust-region subproblem \eqref{eq:TR_subprob} with $\|\bm d_k^\star\|\le r_k$. Then
\begin{equation}\label{eq:nonneg_pred}
\Delta \widetilde{L}_k(\bm z_k,\bm d_k^\star)
\;=\;J(\bm z_k)-\widetilde{L}_k(\bm d_k^\star)\ \ge\ 0 .
\end{equation}
If equality holds in \eqref{eq:nonneg_pred}, then $\bm 0\in \partial \widetilde L_k(\bm 0)$ and
\begin{equation}\label{eq:approx_stationarity_bound}
\mathrm{dist}\big(\bm 0,\partial J(\bm z_k)\big)
\;\le\;
\|\bm\lambda\|\,\|\widetilde{\mathbf S}_k-\mathbf S_k\|.
\end{equation}
In particular, if $\widetilde{\mathbf S}_k=\mathbf S_k$, then $\bm 0\in\partial J(\bm z_k)$, i.e., $\bm z_k\in\mathcal T$.
\end{lem}

\begin{proof}
First, the zero step $\bm d=\bm 0$ is feasible for \eqref{eq:TR_subprob}. By construction of the linearized model at $\bm z_k$,
$\widetilde L_k(\bm 0)=J(\bm z_k)$. Since $\bm d_k^\star$ minimizes $\widetilde L_k$
over $\{\bm d:\ \|\bm d\|\le r_k\}$,
\[
\widetilde L_k(\bm d_k^\star)\ \le\ \widetilde L_k(\bm 0)\ =\ J(\bm z_k),
\]
and \eqref{eq:nonneg_pred} follows.

Assume now that equality holds in \eqref{eq:nonneg_pred}. Then $\widetilde L_k(\bm d_k^\star)=\widetilde L_k(\bm 0)$,
so $\bm 0$ is also a minimizer of the convex subproblem. Since $\bm 0$ lies in the interior of the trust region,
the first--order optimality condition gives $\bm 0\in\partial \widetilde L_k(\bm 0)$.

To relate $\partial \widetilde L_k(\bm 0)$ and $\partial J(\bm z_k)$, write $\widetilde L_k(\bm d)=\ell_k(\bm d)+\phi_{\widetilde{\mathbf S}_k}(\bm d)$, 
$L_k(\bm d)=\ell_k(\bm d)+\phi_{\mathbf S_k}(\bm d)$, where $\ell_k(\bm d):=c(\bm z_k+\bm d)+\bm\eta^\top[h(\bm z_k+\bm d)]_+$ and
$\phi_{\mathbf M}(\bm d):=\bm\lambda^\top|\bm d_y-\mathbf M\bm d_u|$ for any matrix $\mathbf M$.
Let $\Xi:=\{\bm\xi\in\mathbb R^{n}:\ |\xi_i|\le \lambda_i,\ i=1,\dots,n\}$. Then, a direct characterization of the generalized differential at $\bm d=\bm 0$ yields
\vspace{-0.2cm}
\[
\partial \phi_{\mathbf M}(\bm 0)
=
\left\{
\begin{bmatrix}
\bm\xi\\
-\mathbf M^\top\bm\xi
\end{bmatrix}
:\ \bm\xi\in\Xi
\right\}.
\]

Since $\ell_k$ is identical in $L_k$ and $\widetilde L_k$, we have $\partial \widetilde L_k(\bm 0)=\partial \ell_k(\bm 0)+\partial \phi_{\widetilde{\mathbf S}_k}(\bm 0)$, $\partial L_k(\bm 0)=\partial \ell_k(\bm 0)+\partial \phi_{\mathbf S_k}(\bm 0)$. Because $\bm 0\in\partial\widetilde L_k(\bm 0)$, there exist
$\bm s^\ast\in\partial\ell_k(\bm 0)$ and
$\bm\xi^\ast\in\Xi$ such that
\[
\bm s^\ast
+
\begin{bmatrix}
\bm\xi^\ast\\
-\widetilde{\mathbf S}_k^\top\bm\xi^\ast
\end{bmatrix}
=
\bm 0.
\]

The complete characterization of the generalized differential of
$L_k$ at $\bm 0$ is
\[
\partial L_k(\bm 0)
=
\left\{
\bm s+
\begin{bmatrix}
\bm\xi\\
-\mathbf S_k^\top\bm\xi
\end{bmatrix}
:
\bm s\in\partial\ell_k(\bm 0),\
\bm\xi\in\Xi
\right\}.
\]
Using the same pair $(\bm s^\ast,\bm\xi^\ast)$, we construct the
following element of $\partial L_k(\bm 0)$:
\[
\begin{aligned}
\bm v
=
\bm s^\ast+
\begin{bmatrix}
\bm\xi^\ast\\
-\mathbf S_k^\top\bm\xi^\ast
\end{bmatrix}=
\begin{bmatrix}
\bm 0\\
(\widetilde{\mathbf S}_k-\mathbf S_k)^\top\bm\xi^\ast
\end{bmatrix}.
\end{aligned}
\]
It follows that $\|\bm v\|
\le
\|\widetilde{\mathbf S}_k-\mathbf S_k\|\,
\|\bm\xi^\ast\|
\le
\|\bm\lambda\|\,
\|\widetilde{\mathbf S}_k-\mathbf S_k\|$, where the last inequality follows from
$\bm\xi^\ast\in\Xi$.

Finally, since $\bm e(\bm z)=\bm y-\bm g(\bm u)$ is continuously differentiable at $\bm z_k$ with Jacobian
$[\mathbf I,\,-\mathbf S_k]$, the Clarke chain rule implies that the generalized differential of the
equality penalty term in $J$ at $\bm z_k$ coincides with that of $\phi_{\mathbf S_k}$ at $\bm d=\bm 0$;
combined with the identical $\ell_k$ part, this yields $\partial L_k(\bm 0)=\partial J(\bm z_k)$.
Therefore, $\mathrm{dist}\big(\bm 0,\partial J(\bm z_k)\big)
=
\mathrm{dist}\big(\bm 0,\partial L_k(\bm 0)\big)
\le \|\bm v\|
\le \|\bm\lambda\|\,\|\widetilde{\mathbf S}_k-\mathbf S_k\|$, which proves \eqref{eq:approx_stationarity_bound}. The special case $\widetilde{\mathbf S}_k=\mathbf S_k$
gives $\bm 0\in\partial J(\bm z_k)$.
\end{proof}

Lemma~\ref{lem:pred_nonneg} shows that the decrease measured by the data-driven linearized model produced by
\eqref{eq:TR_subprob} is always nonnegative. If the algorithm terminates in finitely many iterations because
$\Delta \widetilde L_k(\bm z_k,\bm d_k^\star)=0$, then \eqref{eq:approx_stationarity_bound} implies that the termination point lies in an
$O(\|\widetilde{\mathbf S}_k-\mathbf S_k\|)$-neighborhood of $\mathcal T$. In the case where $\widetilde{\mathbf S}_k=\mathbf S_k$,
this residual vanishes and $\bm z_k\in\mathcal T$; with feasibility for \eqref{eq:Optimization}, Theorem~3.9 of~\cite{mao2018successive} implies that $\bm z_k$ is a KKT point of the original problem.

When the proposed algorithm generates an infinite sequence $\{\bm z_k\}$, the next lemma establishes a lower bound
on the ratio $\rho(\bm z_k,r_k)$ for nonstationary points, which will be used to depict the stationarity of every limit point of $\{\bm z_k\}$.

\begin{lem}\label{lem:weak_accept_est}
Let $\bar{\bm z}$ be feasible for \eqref{eq:penalty_cost} and not stationary. Then there exist
$\bar\epsilon>0$, $\bar r>0$, and $\kappa>0$ such that for any $\bm z_k\in N(\bar{\bm z},\bar\epsilon)$ and any $r_k\in(0,\bar r]$, if
$\bm d_k^\star$ minimizes $\widetilde L_k(\bm d)$ over $\{\bm d:\ \|\bm d\|\le r_k\}$, then
\[
\rho(\bm z_k,r_k)
:=\frac{\Delta J_k(\bm z_k,\bm d_k^\star)}{\Delta \widetilde L_k(\bm z_k,\bm d_k^\star)}
\;\ge
1-
\frac{C_Mr_k}
{\kappa/2-C_Mr_k}.
\]
\end{lem}

The detailed proof of this lemma generally follows arguments in Lemma 3.11 of~\cite{mao2018successive} and is provided in Appendix~\ref{app:lem_weak_accept_est} for completeness. 
Lemma~\ref{lem:weak_accept_est} implies that, for sufficiently small trust-region, solving the convex linearized subproblem~\eqref{eq:TR_subprob} yields steps whose performance closely reflects the behavior of the original objective, thereby justifying the use of data-driven successive linearization. This result underpins the proof of
the next lemma,
% the Theorem~\ref{thmmain}, 
which establishes convergence to a neighborhood of the stationary point when the proposed algorithm generates an infinite sequence $\{\bm z_k\}$.

\begin{lem}\label{lem:limit points}
Let Assumptions~\ref{assump:fixed_exogenous}--\ref{assump:cost} hold,
then the data-driven successive linearization algorithm always has limit points.
Let $\bar{\bm z}$ be any limit point and $\{\bm z_{k_i}\}$ be a subsequence converging to it, then $\bar{\bm z}$ satisfies the stationarity bound for the penalty problem~\eqref{eq:penalty_cost}: 
\[\mathrm{dist}\big(\bm 0,\partial J(\bar{\bm z})\big)
\ \le\ \|\bm\lambda\|\,\liminf_{i\to\infty}\|\widetilde{\mathbf S}_{k_i}-\mathbf S_{k_i}\|.\] 

\end{lem}

The proof is established by contradiction, showing that any limit point cannot lie outside the corresponding neighborhood of the stationary set of the penalized problem. The detailed proof is given in Appendix~\ref{app:lem_limit points}.

 The above results establish convergence to a neighborhood of the KKT point.  In particular,
 when the Jacobian estimation error vanishes, $\bar{\bm z}$ is a stationary point of~\eqref{eq:penalty_cost}, i.e.,
$\|\widetilde{\mathbf S}_{k_i}-\mathbf S_{k_i}\|\to 0$ implies
$\bar{\bm z}\in\mathcal T$. If $\bar{\bm z}$ is feasible for the original problem~\eqref{eq:Optimization}, it means the iterates converge to an exact KKT point of~\eqref{eq:Optimization}. This concludes the proof for Theorem~\ref{thmmain}. To align with the convergence analysis, the trust region should be sufficiently small while being selected jointly with the sensitivity estimation window to provide sufficient excitation. In practice, the trust region radius can be adjusted based on the magnitude of recent control increments and the magnitude and variation of the estimated sensitivity matrix, so as to balance sufficient excitation for Jacobian estimation with local linearization accuracy.
In addition, to make $\|\widetilde{\mathbf S}_{k_i}-\mathbf S_{k_i}\|$ sufficiently close to zero, whenever the algorithm stalls (e.g., the update $\|\bm z_{k+1}-\bm z_k\|$ falls below a prescribed tolerance), one could choose to inject small zero-mean random perturbations into the control input and then recompute $\widetilde{\mathbf S}_k$. The magnitude of the injected perturbations should be chosen to provide sufficient excitation for Jacobian estimation without excessively degrading voltage regulation performance. More detailed analysis regarding the impact of measurement noise as well as the injected action perturbation, are presented in the Appendix~\ref{app:measurement_noise} and~\ref{app:perturbation_excitation}.

%% file: simulation_journal.tex
To validate the performance of the proposed data-driven successive linearization method for voltage control, we conduct experiments on the IEEE 33-bus system \cite{baran1989network} with  both time-invariant and time-varying net load. All electrical quantities are expressed in per unit (p.u.) on a 100 MVA, 12.66 kV base. 
Each bus has controllable reactive power within $\pm0.05$ p.u., with a voltage reference of $1.00$~p.u. 
At each step we solve the DistFlow~\cite{chiang2002existence} to obtain voltage magnitudes resulting from the current active and reactive power injections. The cost value at each time $t$ is the quadratic voltage-tracking error plus a penalty on the control effort:
$
c (\mathbf{y}_t, \bm{u}_t)  = 0.5\bigl\|\mathbf{y}_t-\mathbf{1}\bigr\|^2+0.1\bigl\|\bm{u}_t\bigr\|^2$,  where $\mathbf{1}$ denotes the all-one vector. To evaluate voltage regulation performance across the entire network, the voltage tracking error in this section includes all buses, regardless of whether their voltages are observed by the controller. We compare the proposed data–driven successive linearization (DDSL) with the following voltage control methods:
\newcommand{\StaticControllerBlock}[2]{%
  \begin{minipage}[t]{0.49\linewidth}
    \centering
    \begin{minipage}[t]{0.49\linewidth}
      \centering
      \includegraphics[width=\linewidth]{figures_static/#1_voltages_static_9nodes.pdf}
    \end{minipage}\hfill
    \begin{minipage}[t]{0.49\linewidth}
      \centering
      \includegraphics[width=\linewidth]{figures_static/#1_actions_static_9nodes.pdf}
    \end{minipage}\\[-0.45em]
    \parbox[t]{0.98\linewidth}{\centering\scriptsize #2: voltage (left) and action (right).}
  \end{minipage}%
}

\begin{figure*}[t]
  \centering

  \StaticControllerBlock{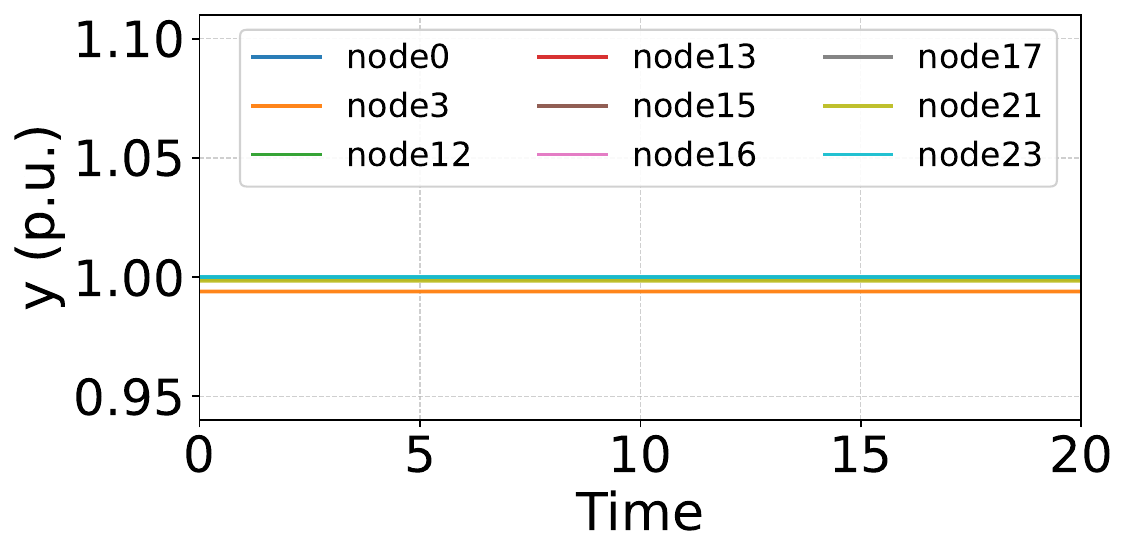}{(a) Convex Relaxation}\hfill
  \StaticControllerBlock{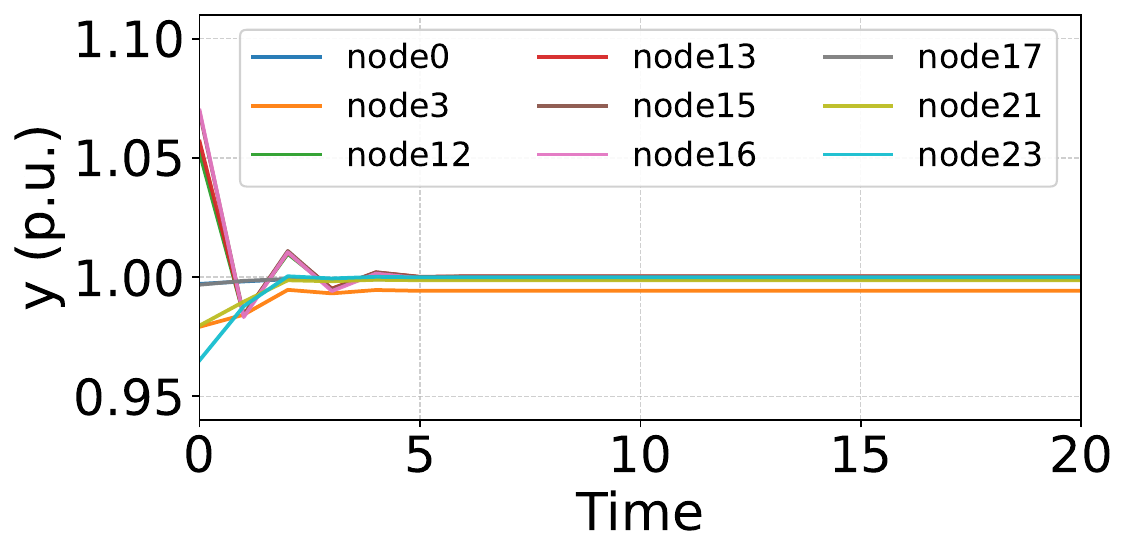}{(b) Data-driven Successive Linearization}\\[0.75em]

  \StaticControllerBlock{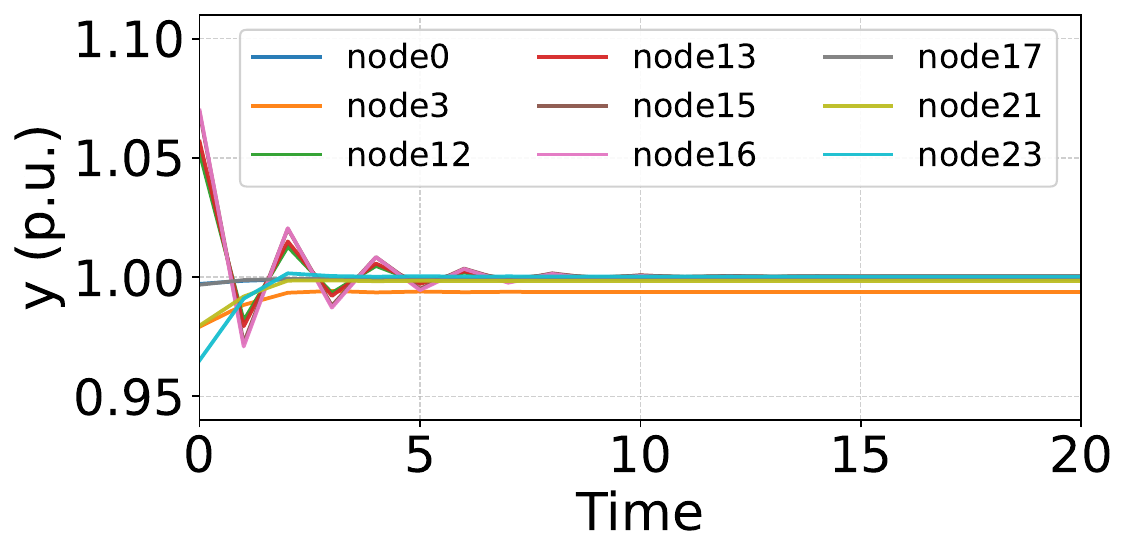}{(c) Fixed Linearization}\hfill
  \StaticControllerBlock{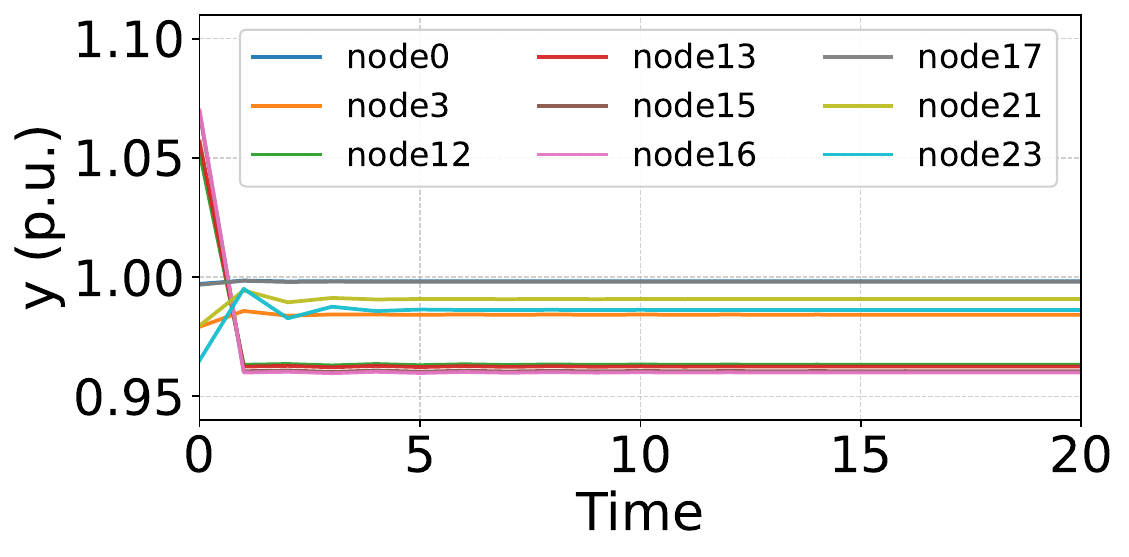}{(d) Linear Control}\\[0.75em]

  \StaticControllerBlock{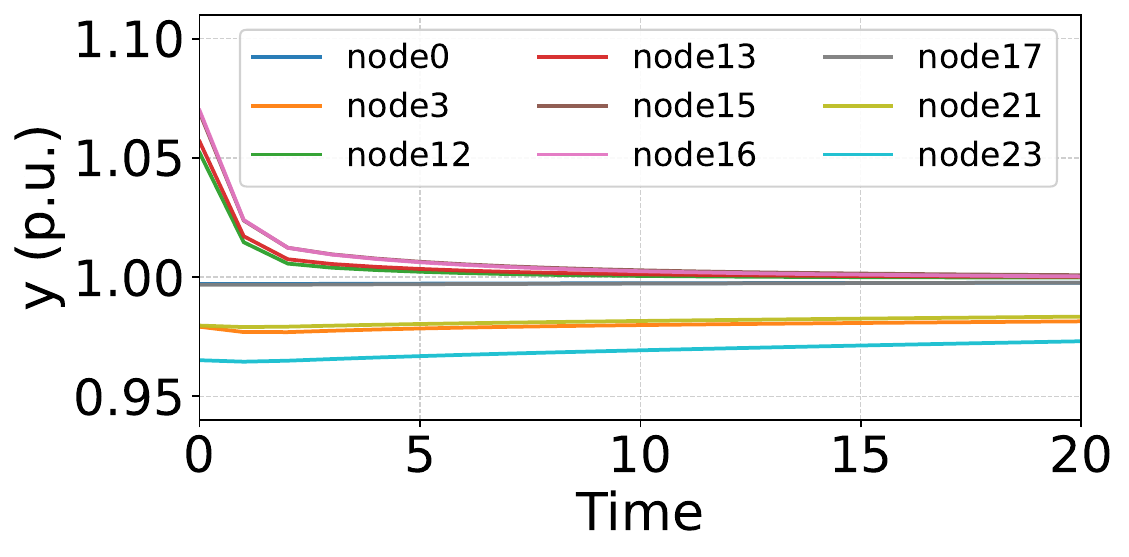}{(e) Gradient Descent-based Feedback Optimization}\hfill
  \StaticControllerBlock{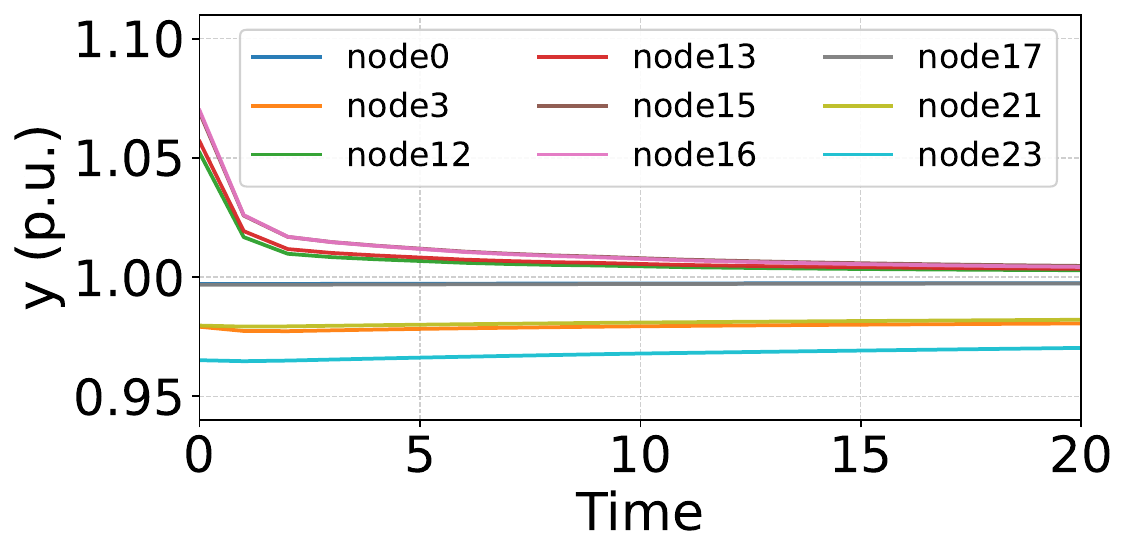}{(f) Primal-dual Feedback Optimization}\\[0.75em]

  \StaticControllerBlock{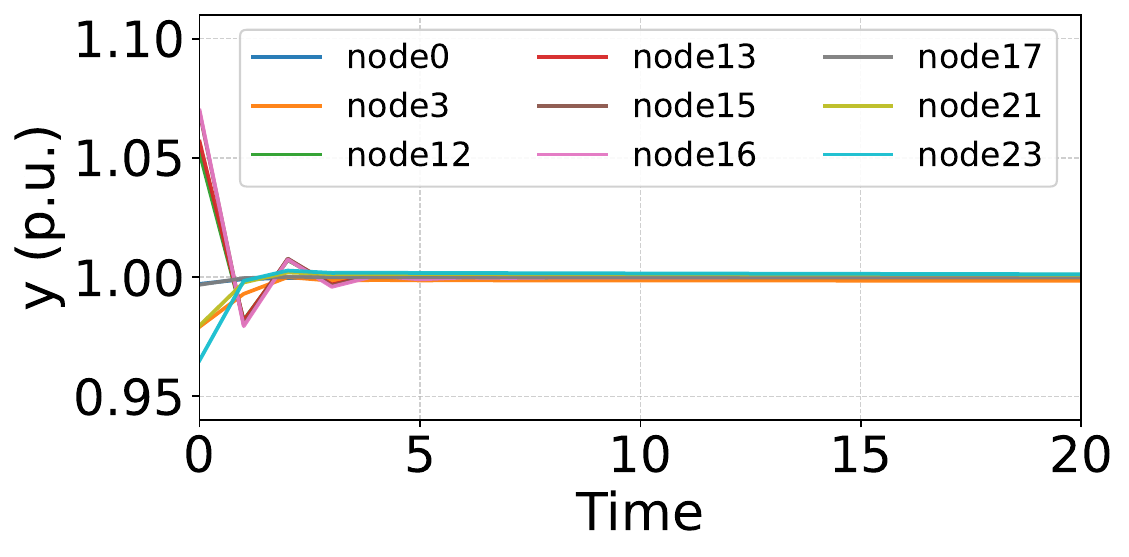}{(g) Projected Gradient Flow}\hfill
  \StaticControllerBlock{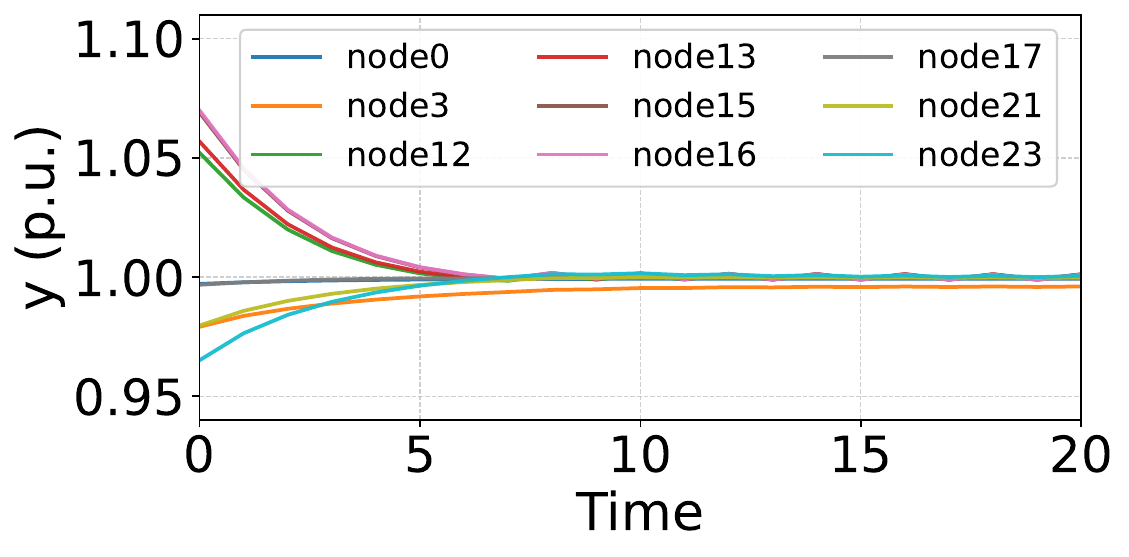}{(h) Feedback-based Safe Gradient Flow}

  \caption{
  Voltage and reactive power trajectories at 9 selected nodes in the time-invariant load setting.
  }
  \vspace{-0.4cm}
  \label{fig:static_traj_all_methods}
\end{figure*}
\begin{itemize}
  \item \textbf{Convex Relaxation (CR)~\cite{farivar2013branch}.} 
  Optimal power flow-based voltage control, which solves formulation in~\eqref{eq:Optimization} using the DistFlow model. Quadratic equality constraints corresponding to line losses are relaxed to be convex inequalities, resulting in a second-order cone program. Details can be found in~\cite{farivar2013branch}.

    \item \textbf{Fixed Linearization (FL).}
  This method uses the same linearized optimization framework as our proposed DDSL but keeps
  the voltage sensitivity matrix fixed throughout the
  control process.

  \item \textbf{Linear Controller (LC)~\cite{photovoltaics2018ieee}.} Linear control law following the IEEE 1547‑2018 standard, expressed as $u_{(i)}= \varphi_i(y_{(i)}-1)$ with $\varphi_i$ being the linear control gain optimized through a learning-based framework in~\cite{cui_leveraging_2025}. This is a fully decentralized controller, and it
coincides with the optimum in special cases~\cite{farivar2013equilibrium}.

\item \textbf{Gradient Descent-Based Feedback Optimization
  (GDFO)~\cite{dominguez2023online}.} Update the control input via projected gradient descent $\bm{u}_{k+1}=\text{Proj}_{\mathcal{U}}(\bm{u}_{k}-\eta \nabla_{ \bm{u}_k}c(\bm{u}_k))$ with $\eta$ being the learning rate~\cite{dominguez2023online}. The gradient of the cost function is computed through the chain rule as $\nabla_{ \bm{u}_k}c(\bm{u}_k)=
\nabla_{ \bm{u}_k}\bm{g}(\bm{u}_k)\nabla_{ \bm{y}_k} c_v(\bm{y}_k)+\nabla_{ \bm{u}_k}c_u(\bm{u}_k)$, where the Jacobian matrix $\nabla_{ \bm{u}_k}\bm{g}(\bm{u}_k)$ is calculated from the online input-output data through~\eqref{eqsenloss}.  

\item \textbf{Primal Dual Feedback Optimization
(PDFO)~\cite{ortmann2020experimental}.}
This method updates the dual variables using measured voltage constraint
residuals as feedback and then determines the control input by minimizing
the corresponding Lagrangian over the control feasible set.

\item \textbf{Projected Gradient Flow
(PGF)~\cite{HAUSWIRTH2024100941}.}
At each iteration, this method projects the negative feedback gradient onto
the feasible set defined by the input constraints and linearized output
constraints. Specifically, it computes
$\bm{\theta}_k=\arg\min_{\bm{\theta}}
\|\bm{\theta}+\eta\nabla_{\bm{u}_k}c(\bm{u}_k)\|^2$
subject to $\bm{u}_k+\bm{\theta}\in\mathcal U$ and
$\bm y_k+\widetilde{\mathbf S}_k\bm\theta\in\mathcal Y$,
and then updates $\bm{u}_{k+1}=\bm{u}_k+\bm{\theta}_k$. The sensitivity matrix $\widetilde{\mathbf S}_k$ is estimated in the same way as DDSL through~\eqref{eqsenloss}.

\item \textbf{Feedback-Based Safe Gradient Flow
(FSGF)~\cite{10777921}.}
This method computes a direction close to the negative feedback gradient while
enforcing control barrier function constraints. Specifically, it solves
$\bm{\theta}_k=\arg\min_{\bm{\theta}}
\|\bm{\theta}+\nabla_{\bm u_k}c(\bm u_k)\|^2$ subject to
$-\widetilde{\mathbf S}_{k,i}\bm{\theta}
\leq-\zeta(\underline y_i-y_{(i),k})$,
$\widetilde{\mathbf S}_{k,i}\bm{\theta}
\leq-\zeta(y_{(i),k}-\overline y_i)$, and the corresponding input safety
constraints, where $\zeta>0$ is the safety parameter, $\underline y_i$ and $\overline y_i$ are the predefined voltage bounds, and
$\widetilde{\mathbf S}_k$ is the Jacobian estimated from online input-output
measurements through~\eqref{eqsenloss}. The control input is then updated as
$\bm u_{k+1}=\bm u_k+\eta\bm{\theta}_k$.

\end{itemize}
In the above methods, the convex relaxation requires exact parameters of the DistFlow model, whereas the others are model-free. We first conduct experiments on a system with time-invariant loads to demonstrate the convergence of the proposed method and the optimality of the cost at convergence. Measurement noise is then added to the voltage measurements to evaluate the robustness of data-driven successive linearization to measurement errors. Subsequently, experiments with time-varying loads are conducted to demonstrate the ability of the proposed method to rapidly adapt to changes in net load. The proposed method also accommodates scenarios in which only a subset of buses is measured and controllable. To evaluate its performance under partial controllability, we use the same time-varying setting and vary the fraction of controllable buses. Actual costs and voltages are reported. The code for this work is publicly available at \url{https://github.com/nyudyw/Data-Driven-Successive-Linearization}.

\subsection{Time-Invariant Load}

Under a fixed load, Fig.~\ref{fig:static_traj_all_methods} demonstrates the voltage and reactive power for 9 selected nodes under all eight methods. The solution of convex relaxation is shown in Fig.~\ref{fig:static_traj_all_methods}(a), which recovers voltages close to 1 p.u.  It is noteworthy that this solution is feasible for the original nonconvex voltage control problem, indicating that this solution is the global optimum of the original nonconvex problem~\eqref{eq:Optimization}, and thus it serves as the model-based benchmark. The proposed data-driven successive linearization in Fig.~\ref{fig:static_traj_all_methods}(b) converges within five control updates, where the voltage and actions at convergence  closely match those produced by convex relaxation. This suggests that the proposed approach reaches a solution very close to the global optimum of the original nonconvex voltage control problem. The fixed linearization method in Fig.~\ref{fig:static_traj_all_methods}(c) attains a similar final state but converges more slowly and exhibits more pronounced oscillations during the initial convergence phase. Among the remaining methods, feedback-based safe gradient flow in Fig.~\ref{fig:static_traj_all_methods}(h) settles closest to these three, while all the other feedback optimization-based methods and linear control exhibit slower convergence and actions at convergence are far from the global optimum.

\begin{figure}[t]
  \centering
\includegraphics[width=0.98\linewidth]{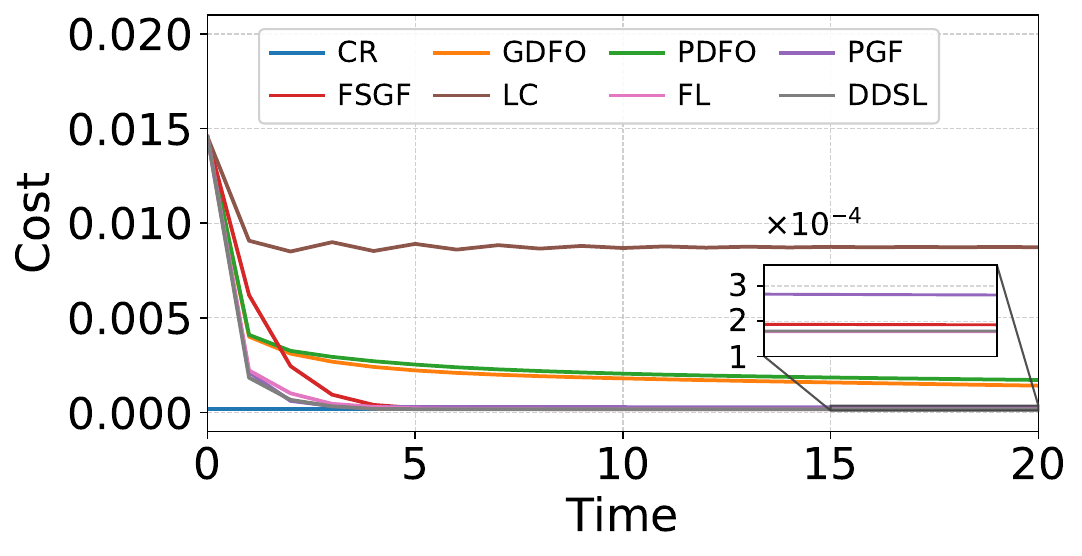}
  \vspace{-0.4cm}
  \caption{Cost trajectories corresponding to Fig.~\ref{fig:static_traj_all_methods}.}
  \vspace{-0.4cm}
\label{fig:coststatic_all_methods}
\end{figure}

Fig.~\ref{fig:coststatic_all_methods} compares the cost functions of all eight methods across iterations. The proposed data-driven successive linearization rapidly drives the cost to the convex relaxation level within a few iterations, indicating fast convergence. The cost achieved by data-driven successive linearization (gray) converges to the global optimum obtained through the convex relaxation (blue). At convergence, the cost achieved by data-driven successive linearization is 88.38\%, 90.25\%, 37.51\%, 9.62\%, and 98.03\% lower than those of gradient descent-based feedback optimization, primal-dual feedback optimization, projected gradient flow, feedback-based safe gradient flow, and linear control, respectively.

To evaluate performance under different load levels, we further collect voltage and control action trajectories for 100 cases with randomly perturbed loads lasting 20 steps. Fig.~\ref{fig:tail5_all_methods} summarizes the mean costs over the last five steps of the trajectories under different methods. For visualization purposes, all costs are normalized using the same basis and displayed on a logarithmic y-axis. The proposed data-driven successive linearization achieves the lowest mean cost across the 100 cases, which is 83.76\%, 23.03\%, 91.31\%, 93.00\%, 49.16\%, 17.86\%, and 99.40\% lower than those achieved by convex relaxation, fixed linearization, gradient descent-based feedback optimization, primal-dual feedback optimization, projected gradient flow, feedback-based safe gradient flow, and linear control, respectively.  Interestingly, the proposed approach can even achieve lower cost than convex relaxation. This occurs because the relaxation may not be exact when the system is lightly loaded, so the solution from convex relaxation can incur higher cost when applied to the system using the original DistFlow model. Consequently, the proposed approach consistently outperforms feedback optimization methods and linear control, while performing comparable or even better than the model-based voltage control approach.

\begin{figure}[t]
  \centering
  \includegraphics[width=0.95\linewidth]{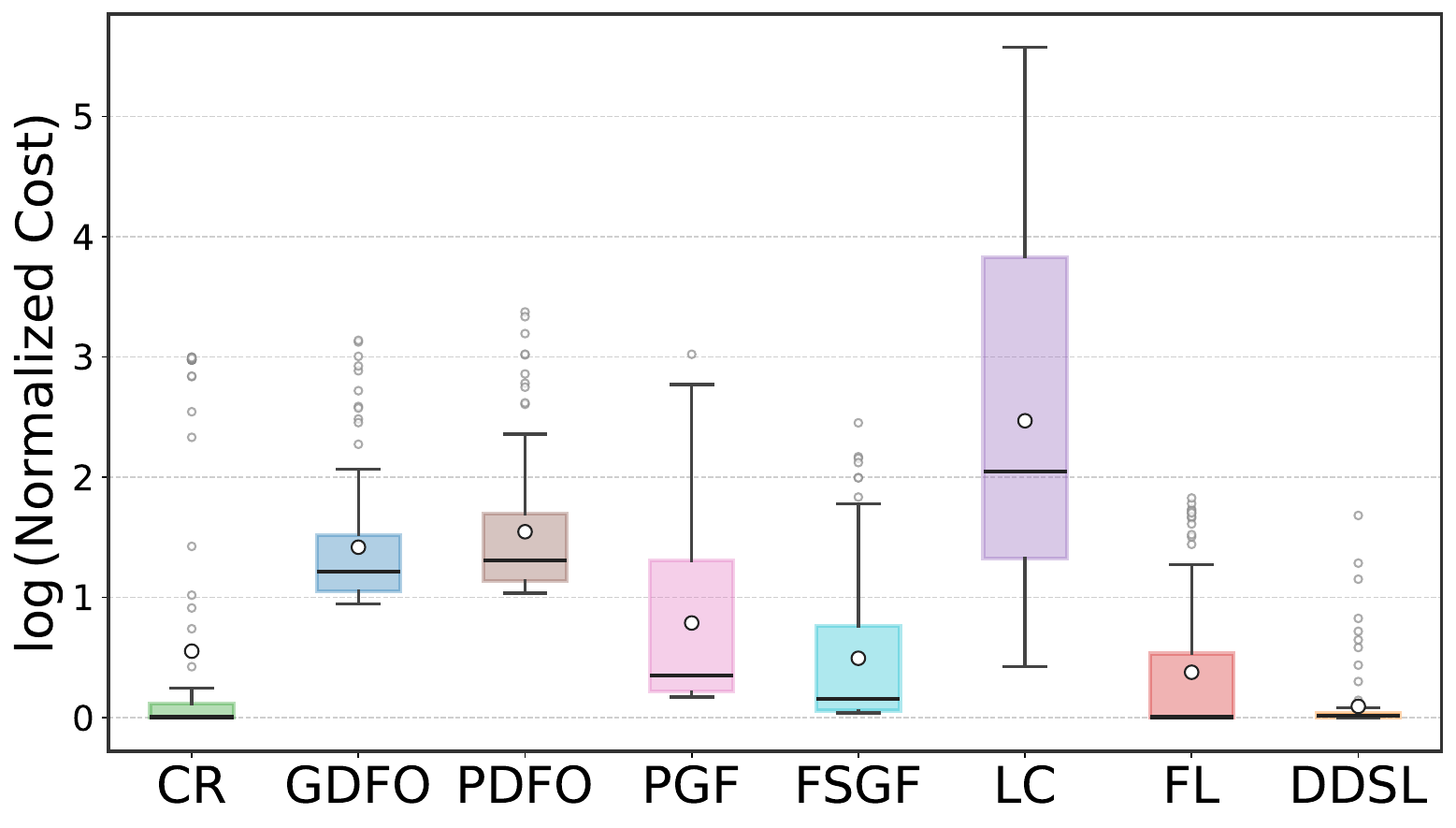}
  \vspace{-0.2cm}
  \caption{Comparison of log-normalized mean cost of the last 5 steps in the time-invariant load setting over 100 trials.}
  \vspace{-0.4cm}
  \label{fig:tail5_all_methods}
\end{figure}

The robustness of data-driven successive linearization to measurement noise is further evaluated in the time-invariant load setting by adding independent zero-mean Gaussian noise to the voltage measurements, with standard deviations ranging from \(0\) to \(0.008\) p.u. at all nodes. For each noise level, the results are averaged over 10 independent noise realizations, with the same realization applied to all methods for a fair comparison. The cost of each realization is evaluated using the true system voltages and averaged over the trajectory. 

\begin{figure}[t]
  \centering
  \includegraphics[width=0.95\linewidth]{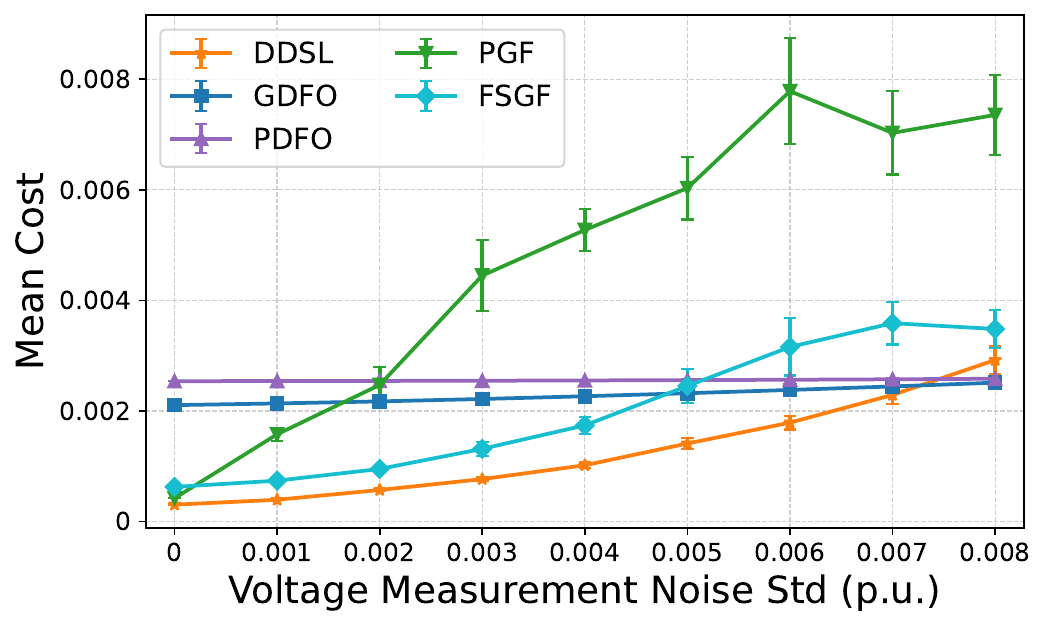}
  \vspace{-0.3cm}
  \caption{Comparison of mean costs under different levels of voltage measurement noise in the time-invariant load setting. Markers and error bars represent the mean and standard error over 10 independent noise realizations, respectively.}
  \vspace{-0.4cm}
  \label{fig:measurement_noise_cost}
\end{figure}

\newcommand{\VaryingControllerBlock}[2]{%
  \begin{minipage}[t]{0.49\linewidth}
    \centering
    \begin{minipage}[t]{0.49\linewidth}
      \centering
      \includegraphics[width=\linewidth]
      {figures_varying/#1_voltages_Varying_8nodes.pdf}
    \end{minipage}\hfill
    \begin{minipage}[t]{0.49\linewidth}
      \centering
      \includegraphics[width=\linewidth]
      {figures_varying/#1_actions_Varying_8nodes.pdf}
    \end{minipage}\\[-0.45em]
    \parbox[t]{0.98\linewidth}{
      \centering\scriptsize #2: voltage (left) and action (right).
    }
  \end{minipage}%
}

\begin{figure*}[t]
  \centering

  \VaryingControllerBlock{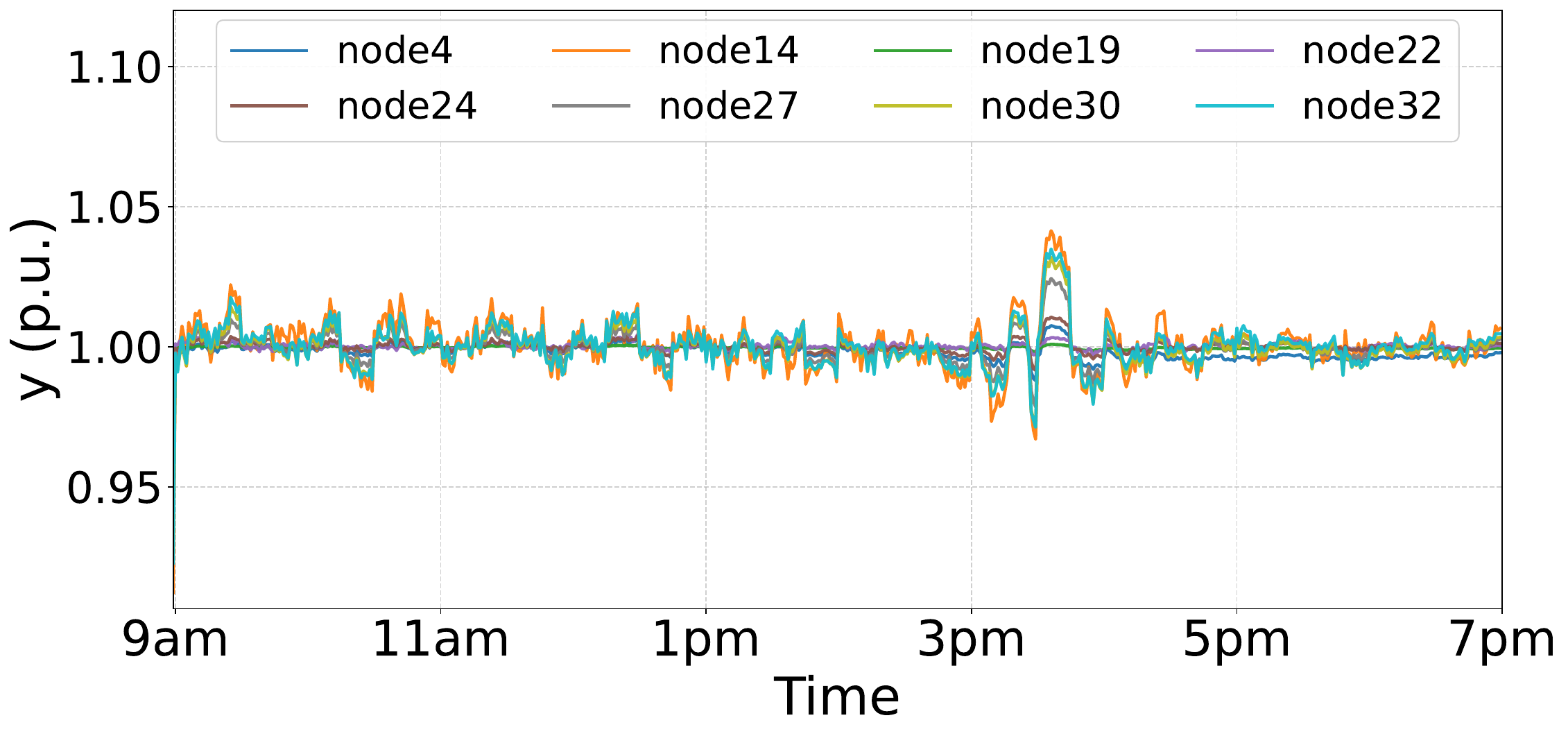}{(a) Convex Relaxation}\hfill
  \VaryingControllerBlock{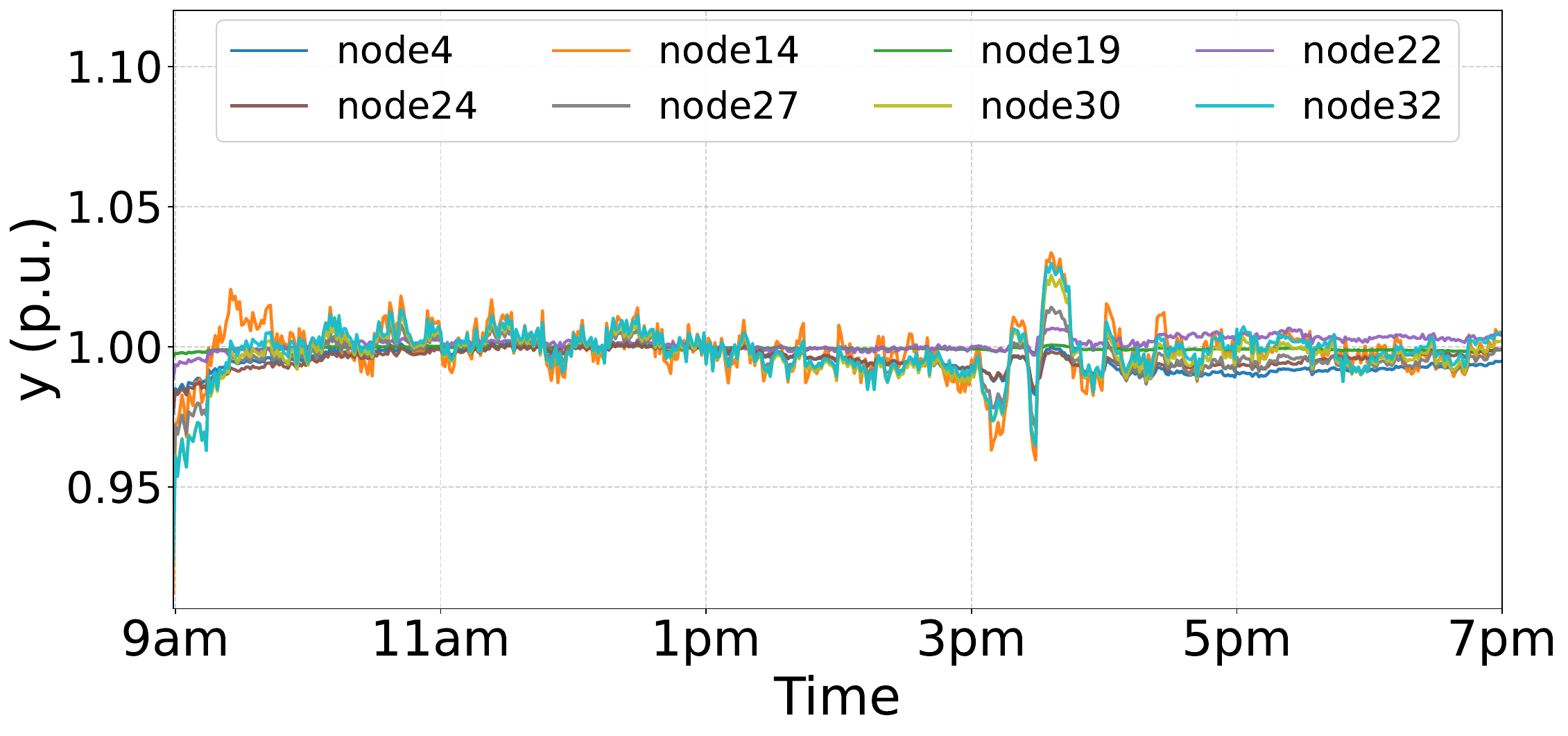}{(b) Data-driven Successive Linearization}\\[0.75em]

  \VaryingControllerBlock{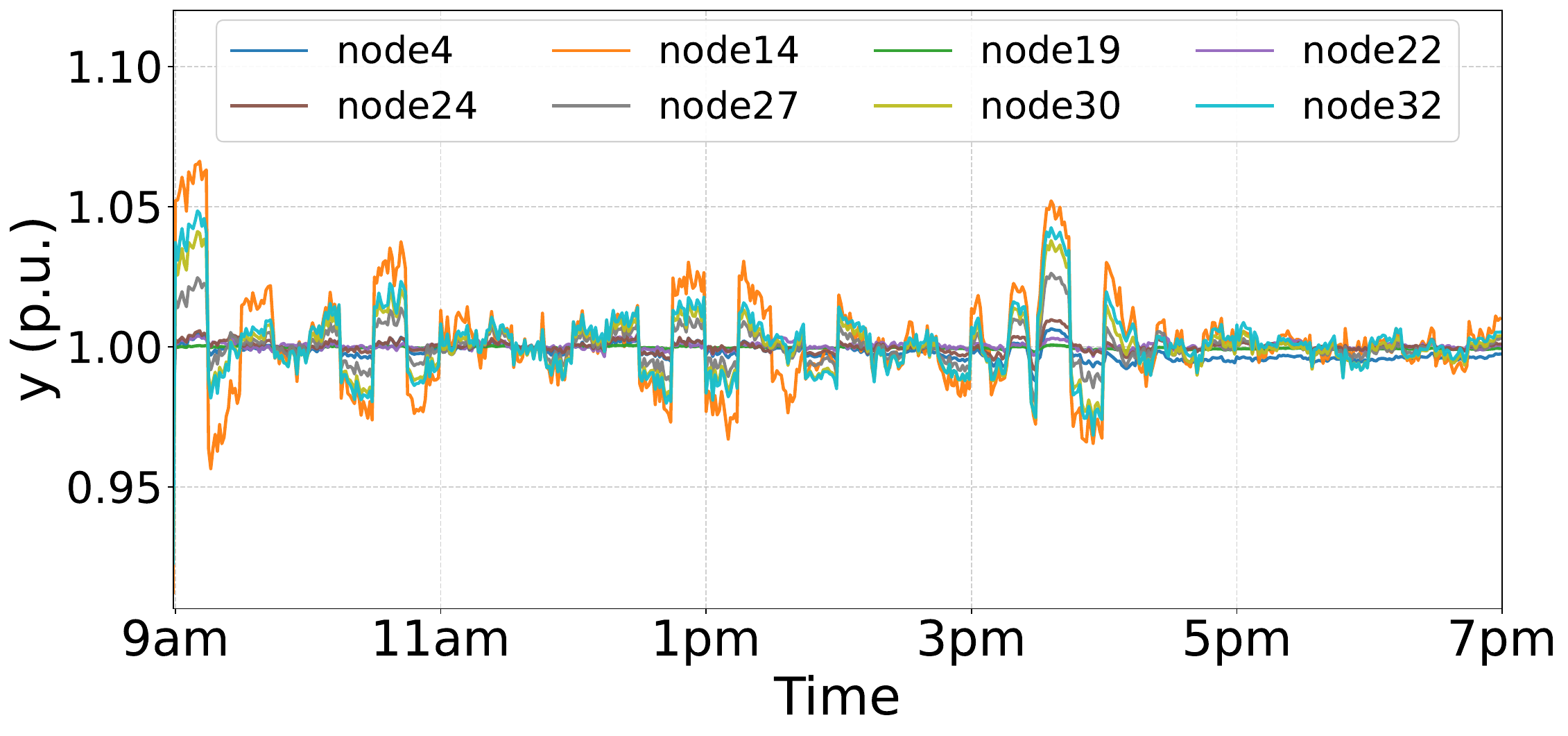}{(c) Fixed Linearization}\hfill
  \VaryingControllerBlock{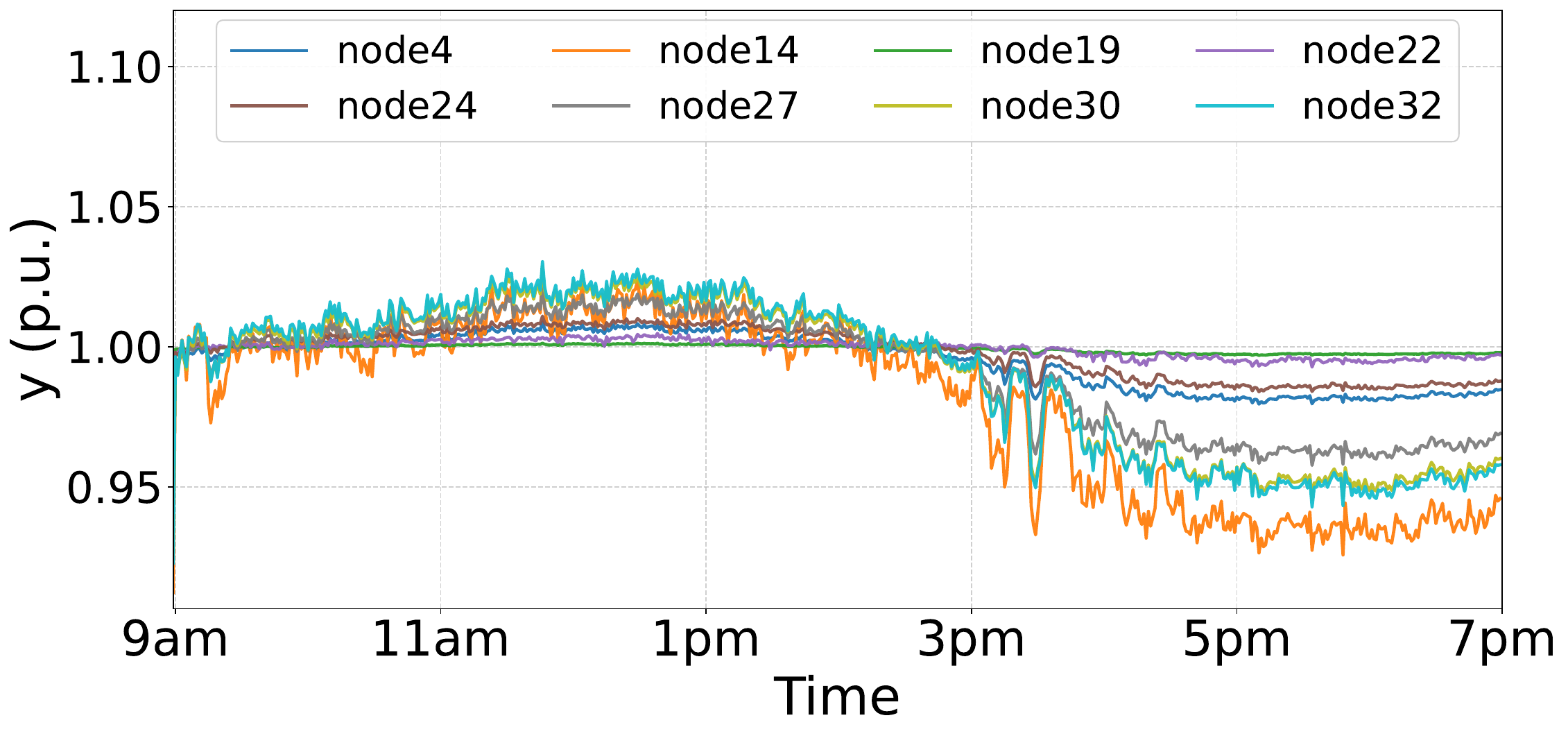}{(d) Linear Control}\\[0.75em]

  \VaryingControllerBlock{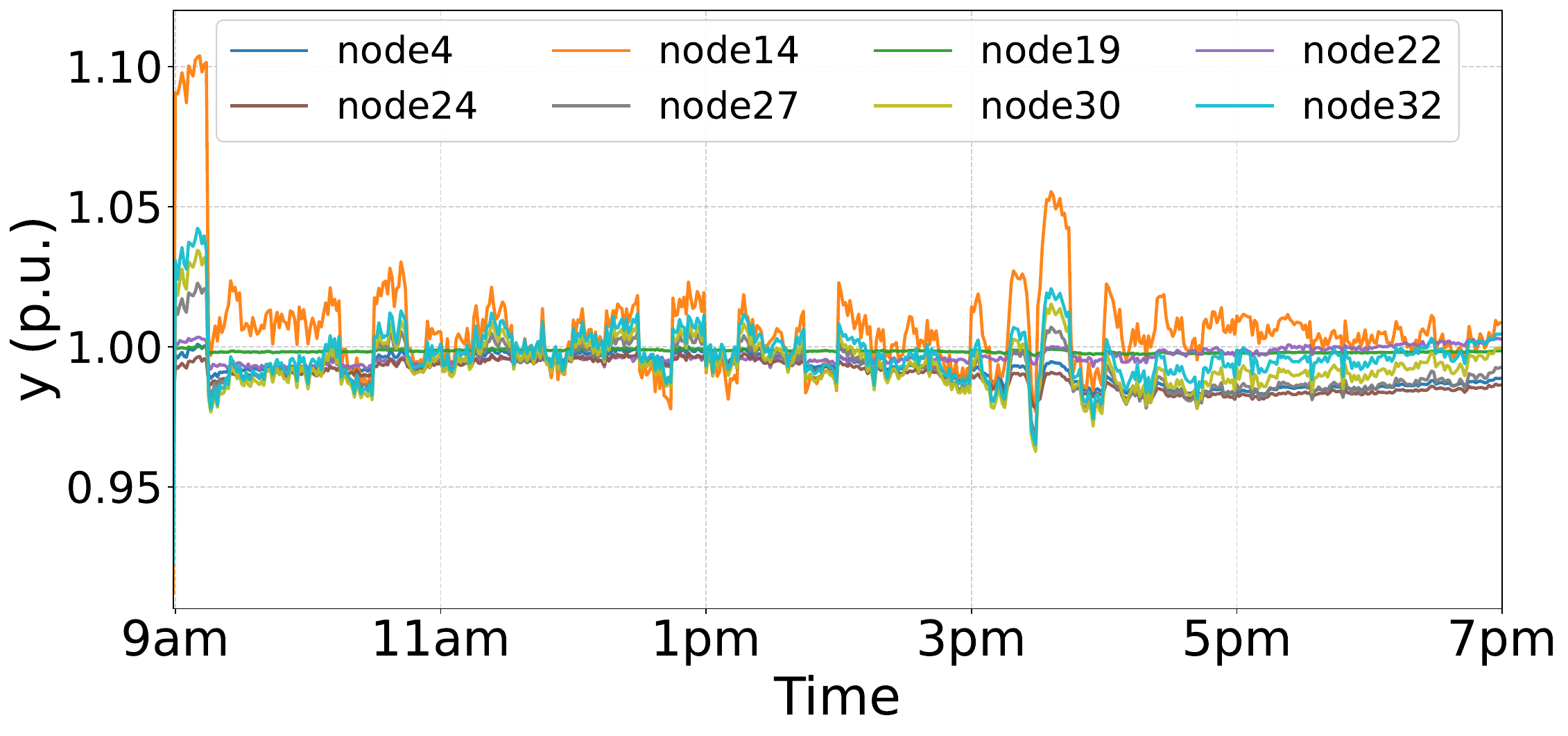}{(e) Gradient Descent-based Feedback Optimization}\hfill
  \VaryingControllerBlock{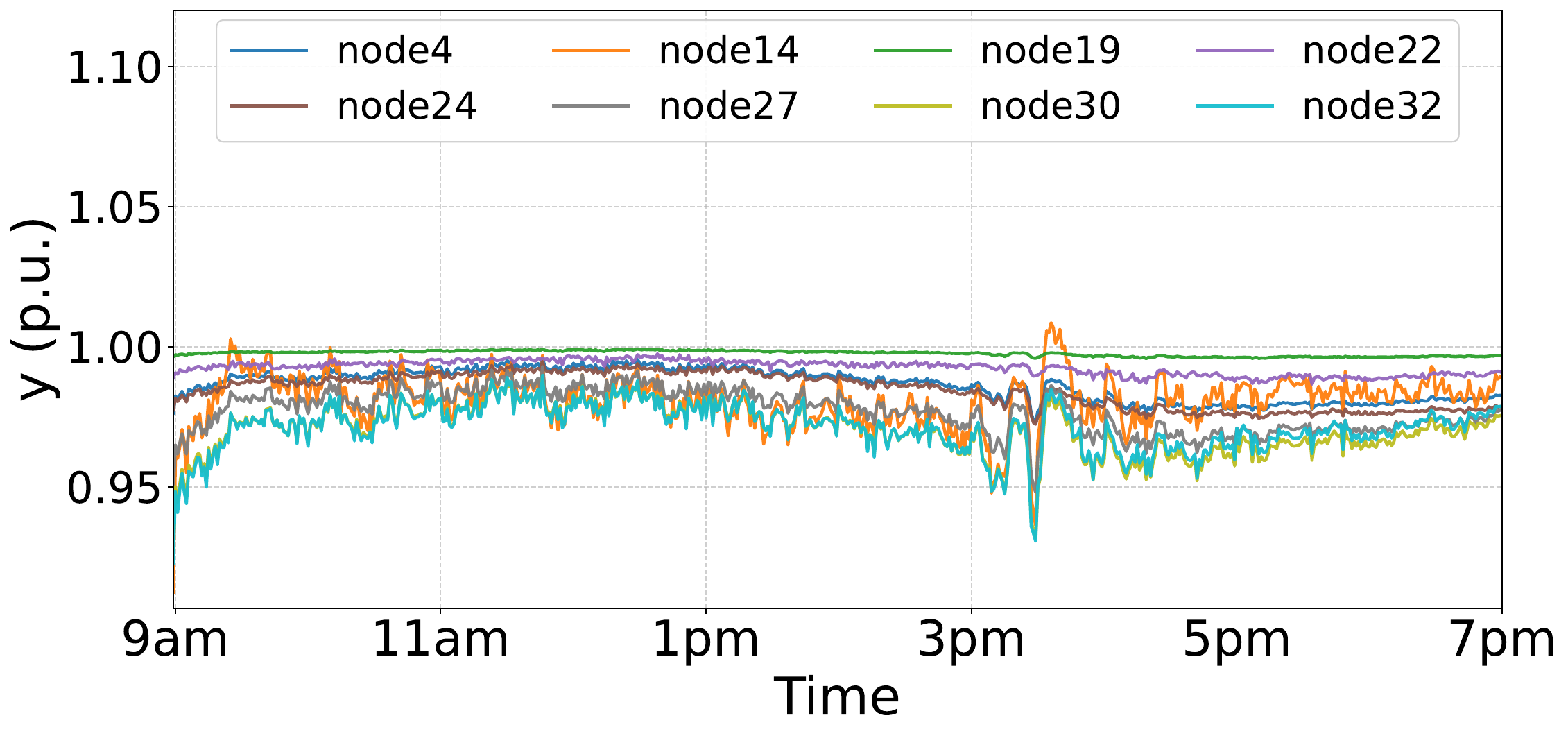}{(f) Primal-dual Feedback Optimization}\\[0.75em]

  \VaryingControllerBlock{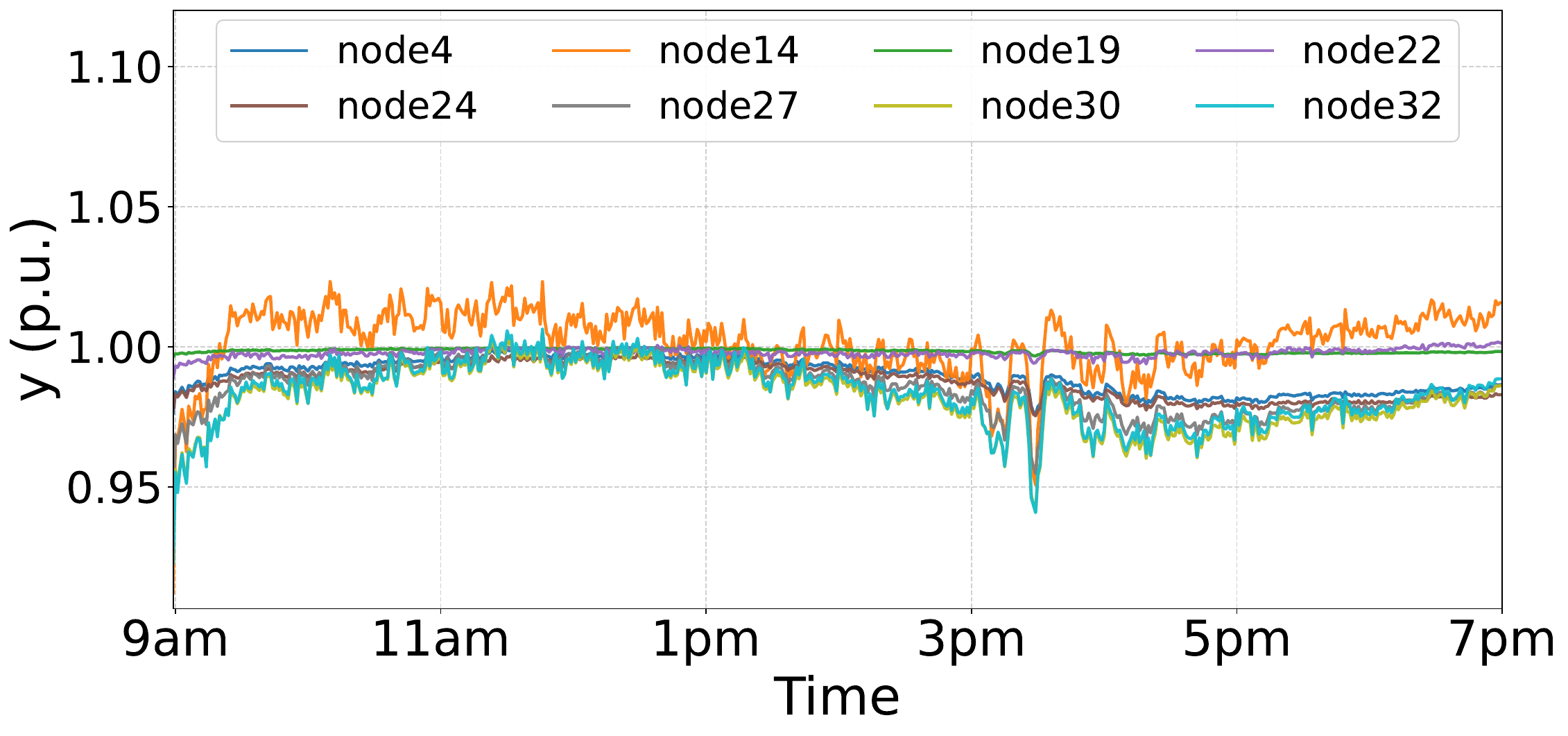}{(g) Projected Gradient Flow}\hfill
  \VaryingControllerBlock{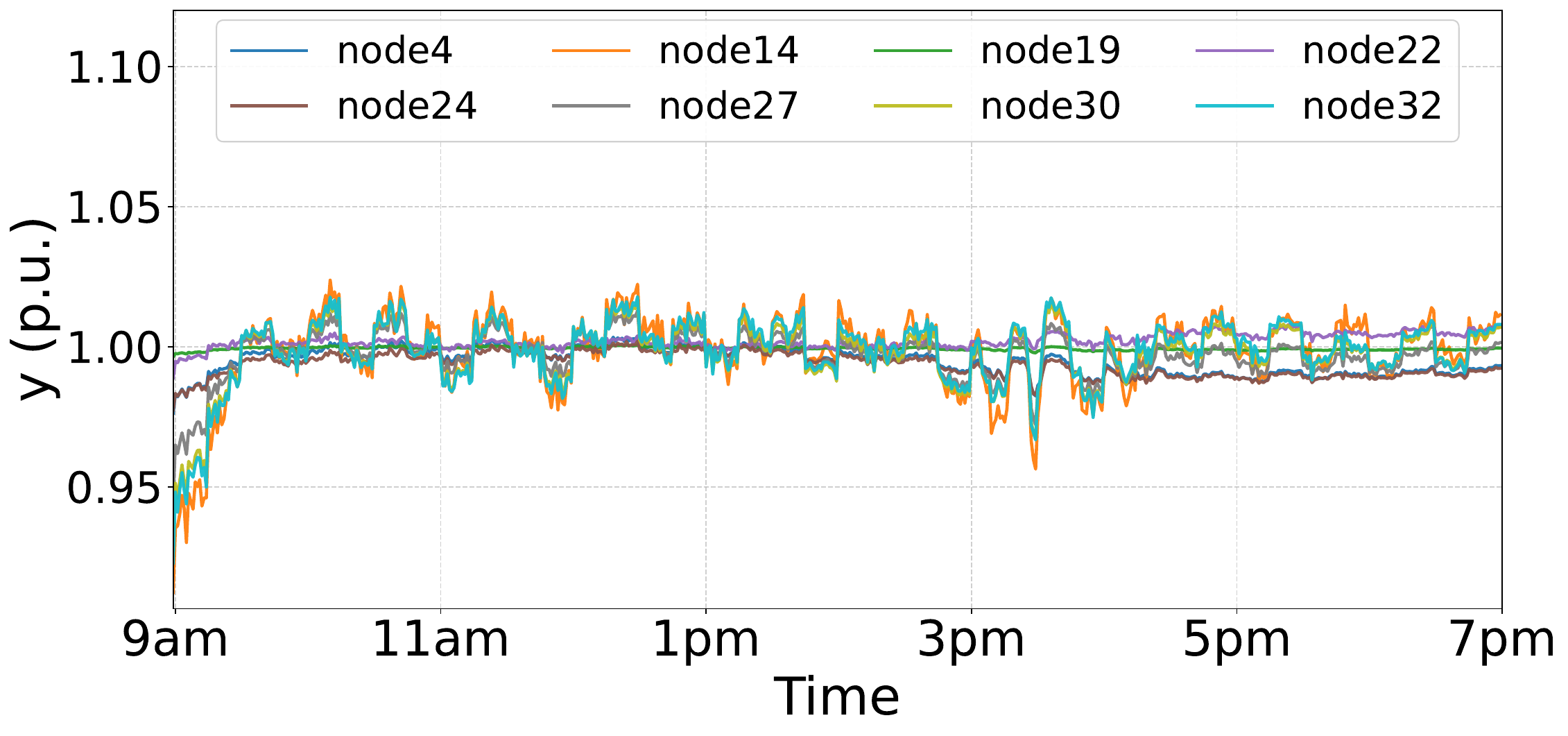}{(h) Feedback-based Safe Gradient Flow}

  \caption{
  Voltage and reactive power trajectories at 8 selected nodes in the
  time-varying load setting.
  }
  \vspace{-0.4cm}
  \label{fig:tv_traj_all_methods}
\end{figure*}

As shown in Fig.~\ref{fig:measurement_noise_cost}, the costs of all controllers generally increase with the measurement-noise level. The proposed data-driven successive linearization achieves the lowest mean cost for noise standard deviations up to \(0.007\) p.u., demonstrating strong robustness over most of the tested range. Projected gradient flow exhibits the largest overall cost increase as the measurement noise grows. Although the costs of gradient descent-based and primal-dual feedback optimization increase more gradually, both methods maintain comparatively high costs across the tested range because they do not fully eliminate voltage deviations from nominal, resulting in persistent voltage-tracking errors. At the highest noise level of \(0.008\) p.u., which is an extreme case in practical power system operation, the mean cost of the proposed method is slightly higher than those of gradient descent-based and primal-dual feedback optimization, but remains \(16.11\%\) and \(60.29\%\) lower than those of feedback-based safe gradient flow and projected gradient flow, respectively. Overall, the proposed method maintains favorable cost performance and satisfactory robustness under moderate measurement noise.

\subsection{Time-Varying Load}

We further evaluate the methods under time-varying load scenarios.
The time-varying load is constructed using minute-resolution power
measurements from an operational distribution grid~\cite{xie2025digital}.
The voltage and reactive power measurements are sampled every minute, and reported to controllers every
15 minutes, at which point each controller updates its reactive power actions. 

Fig.~\ref{fig:tv_traj_all_methods} shows the voltage and reactive power
trajectories on one day.
Convex relaxation and data-driven successive linearization in
Fig.~\ref{fig:tv_traj_all_methods}(a) and
Fig.~\ref{fig:tv_traj_all_methods}(b), respectively, maintain all bus
voltages within the prescribed 0.95--1.05~p.u. limits throughout
the period.
Their action trajectories adapt closely to the load variations and produce
prompt voltage corrections at the reporting times.
For example, following the pronounced load change at approximately 09:06,
data-driven successive linearization keeps the voltages within the limits
and reduces the maximum voltage deviation to below \(0.013\)~p.u. at the
next update at 09:15.

In contrast, fixed linearization in
Fig.~\ref{fig:tv_traj_all_methods}(c) produces larger voltage fluctuations
and occasional voltage violations.
Linear control in Fig.~\ref{fig:tv_traj_all_methods}(d) changes its actions
only gradually and exhibits drastic voltage deviations under
changing loads.
Gradient descent-based feedback optimization in
Fig.~\ref{fig:tv_traj_all_methods}(e) experiences pronounced voltage
overshoots and requires more time to recover after large load variations.
Primal-dual feedback optimization, projected gradient flow, and
feedback-based safe gradient flow in
Fig.~\ref{fig:tv_traj_all_methods}(f)--(h) reduce the large deviations but still
exhibit persistent over- or under- voltage in some periods.
 This experiment highlights the capability of data-driven successive linearization to quickly restore voltages under dynamic load conditions. 

\begin{figure}[t]
  \centering
  \includegraphics[width=0.95\linewidth]{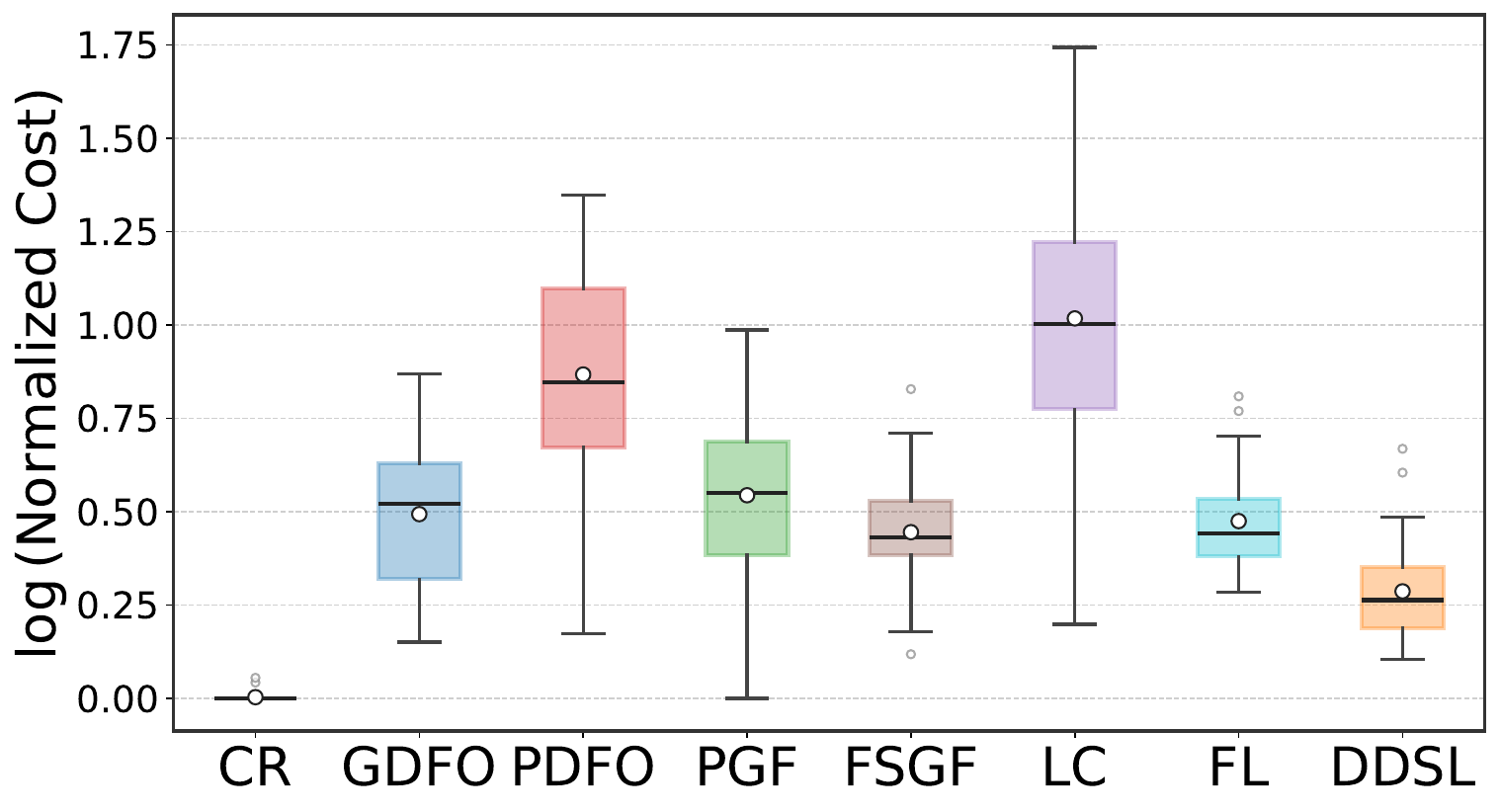}
  \vspace{-0.2cm}
  \caption{Comparison of log-normalized mean costs in the time-varying load setting over 30 days.}
  \vspace{-0.4cm}
  \label{fig:cost_varying_all_methods}
\end{figure}

Across 30 time-varying load cases,
Fig.~\ref{fig:cost_varying_all_methods} compares the costs achieved by
the eight controllers.
Convex relaxation achieves the lowest mean cost and performs best on
28 of the 30 days.
Data-driven successive linearization achieves the second-lowest mean cost
and ranks second on 21 days.
Its mean cost is \(35.24\%\), \(75.73\%\), \(26.17\%\), \(64.57\%\),
\(28.03\%\), and \(29.15\%\) lower than those of fixed linearization,
linear control, gradient descent-based feedback optimization, primal-dual
feedback optimization, projected gradient flow, and feedback-based safe
gradient flow, respectively.
These results demonstrate that the proposed method can quickly adapt to a wide
range of time-varying load conditions and achieves a lower mean
cost than the other feedback optimization-based controllers.

The proposed framework can also be applied when sensing and reactive power
actuation are available only at a subset of buses. We evaluate this setting
using the same measured time-varying load profile as in the above experiment. Firstly, 20 of the 32 non-slack
buses are set to be both metered and equipped with controllable reactive power
resources. The proposed data-driven successive linearization and all
feedback optimization-based baselines use voltage and reactive power
measurements only from this designated subset and apply control actions only
at the same buses every 15 minutes. Convex relaxation is retained as a full information benchmark: it uses the
full network active and reactive power injections to compute its recommended
setpoints and evaluates all bus voltages, while its control actions are
restricted to the same controlled subset. For a consistent comparison, the
costs of all methods are evaluated offline using voltages at all 32 non-slack
buses.

\begin{figure}[t]
  \centering
  \includegraphics[width=0.95\linewidth]{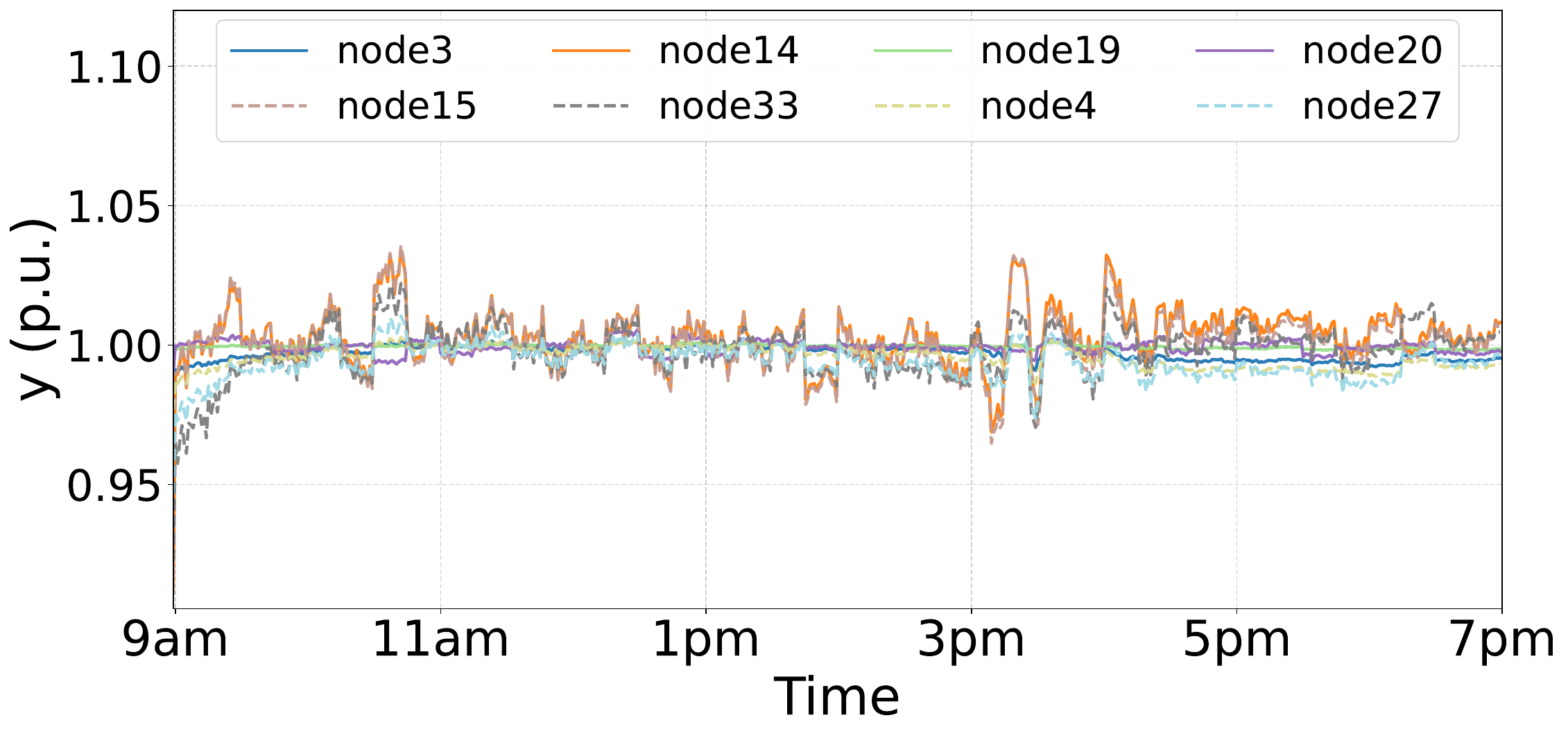}
  \vspace{-0.3cm}
  \caption{Voltage trajectories produced by data-driven successive linearization under time-varying loads when 20 buses are measured and controlled. Solid lines represent selected measured-and-controlled buses, whereas dashed lines represent selected unmeasured-and-uncontrolled buses.}
  \vspace{-0.1cm}
  \label{fig:partial20_ddsl_voltage}
\end{figure}

\begin{figure}[t]
  \centering
  \includegraphics[width=0.95\linewidth]{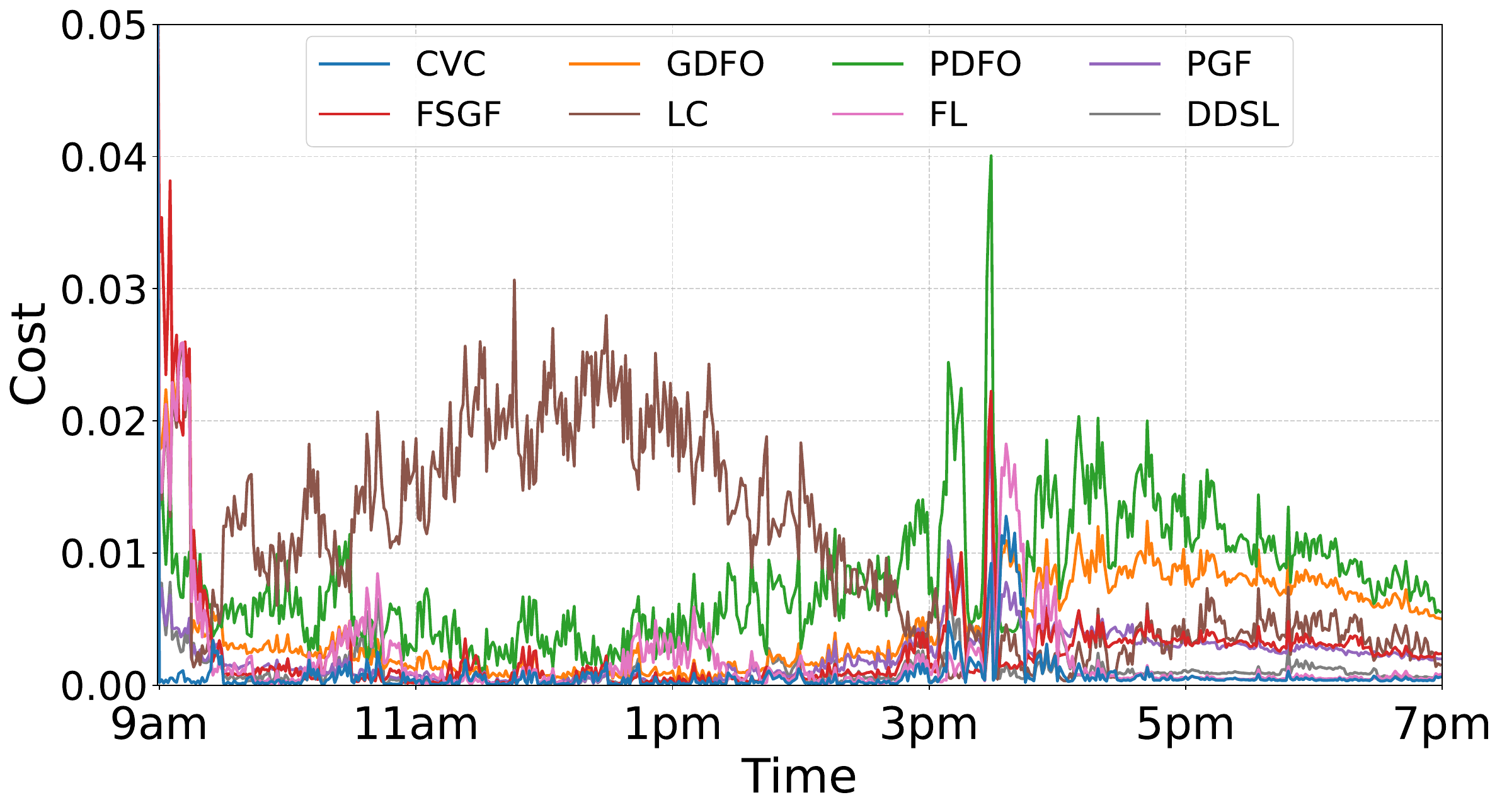}
  \vspace{-0.3cm}
  \caption{Cost trajectories of the eight controllers under time-varying loads when 20 buses are measured and controlled. Convex relaxation is included as a full-information benchmark, while the other controllers use measurements only from the designated subset. All methods apply reactive power actions only at the same 20 controlled buses every 15 minutes.}
  \vspace{-0.4cm}
  \label{fig:partial20_cost}
\end{figure}

As shown in Fig.~\ref{fig:partial20_ddsl_voltage}, the proposed method maintains all measured and controlled buses within the prescribed voltage limits throughout the period. Although the unmeasured buses exhibit larger residual deviations from 1~p.u., their voltages remain within 0.95--1.05~p.u. over the entire period. The cost trajectories in Fig.~\ref{fig:partial20_cost} further show that the proposed method maintains a low full network cost as the load varies over time. Its mean full network cost is close to the full information convex relaxation benchmark and outperforms the other methods, including the linear control that intrinsically requires no communication and
no observability beyond the bus at which it controls. A more restrictive 10-bus case experiment is also conducted and the result is reported in Appendix~\ref{app:partial}. 

Although the convergence analysis is established under fixed exogenous conditions, these experiments demonstrate that the proposed method can effectively adapt to time-varying operating conditions and retain favorable voltage regulation performance even with partial measurement and control availability.

%% file: appendix.tex
\subsection{Proof for Lemma~\ref{lem:est_ptwise}}\label{app:lem_est_ptwise}
% \begin{proof}

% We prove the lemma by exploiting a closed-form representation of the estimated sensitivity matrix
% $\widetilde{\mathbf S}_k$ constructed from the past trajectories of the control inputs $\bm u$ and voltage states $\bm y$. This enables a precise comparison between $\widetilde{\mathbf S}_k$ and the Jacobian $\mathbf S_k$
% by leveraging local Lipschitz continuity and the trust-region bound on recent iterates. 

From the trust region condition $\|\bm z_\ell-\bm z_k\|\le \alpha r_k$ we obtain $\|\bm u_\ell-\bm u_k\|\le \alpha r_k$, and therefore
$\|\Delta\bm u_\ell\|\le \|\bm u_\ell-\bm u_k\|+\|\bm u_{\ell-1}-\bm u_k\|\le 2\alpha r_k$ on the window. As $\Delta\bm y_\ell=\mathbf S_\ell\,\Delta\bm u_\ell+\bm \gamma_\ell$, where $\|\bm \gamma_\ell\|\le \beta\,\|\Delta\bm u_\ell\|^2$
for some constant $\beta>0$, we can write
\begin{equation}\label{eq:Y_decomp}
\bm Y \;=\; \bm U\,\mathbf S_k^\top + \bm E,
\end{equation}
where
$$
\bm E:=
\begin{bmatrix}
\Delta\bm u_{k-\tau+1}^\top(\mathbf S_{k-\tau+1}^\top-\mathbf S_k^\top)+\bm \gamma_{k-\tau+1}^\top\\[-.25em]
\vdots\\[-.25em]
\Delta\bm u_{k}^\top(\mathbf S_{k}^\top-\mathbf S_k^\top)+\bm \gamma_{k}^\top
\end{bmatrix}.
$$
For the weighted least–squares problem, the estimator satisfies $\widetilde{\mathbf S}_k^\top
=
({\bm U}^\top \bm W \bm U)^{-1}{\bm U}^\top \bm W \bm Y$,
where $\bm W\in\mathbb R^{\tau\times\tau}$ is the diagonal weight matrix defined in~\eqref{eq:Wk_diag}.
Substituting \eqref{eq:Y_decomp} into this expression yields
\begin{equation}\label{eq:S_diff_exact}
\widetilde{\mathbf S}_k^\top-\mathbf S_k^\top
\;=\; ({\bm U}^\top \bm W \bm U)^{-1}{\bm U}^\top \bm W \bm E.
\end{equation}

Note that when the operating point
is not near the voltage-collapse boundary, the Jacobian $\nabla_{ \bm{u}}\bm{g}(\bm{u})$ is locally Lipschitz. 
Denote the local Lipschitz constant of the Jacobian as $L_S$. We have
$\|\mathbf S_\ell-\mathbf S_k\|\le L_S\,\|\bm u_\ell-\bm u_k\|\le L_S\,\alpha r_k$; together with
$\|\Delta\bm u_\ell\| \le 2\alpha r_k$, this implies $\big\|\Delta\bm u_\ell^\top(\mathbf S_\ell^\top-\mathbf S_k^\top)\big\|\le\|\Delta\bm u_\ell\|\,\|\mathbf S_\ell-\mathbf S_k\| \le (2\alpha r_k)(L_S\,\alpha r_k)
=
2L_S\alpha^2 r_k^2$.
Moreover, from the remainder bound $\|\bm \gamma_\ell\|\le \beta\|\Delta\bm u_\ell\|^2$, we obtain $\|\bm \gamma_\ell\|\le \beta(2\alpha r_k)^2 = 4\beta\alpha^2 r_k^2$. Hence each row of $\bm E$ has norm at most
$2L_S\alpha^2 r_k^2+4\beta\alpha^2 r_k^2\le c_E' r_k^2$ for some constant $c_E'>0$, and therefore
$\|\bm E\|_F^2 \le \tau (c_E')^2 r_k^4$, so $\|\bm E\|_F \le \sqrt{\tau}\,c_E' r_k^2 \le c_E r_k^2$ for some constant $c_E>0$.
In particular, $\|\bm E\|\le \|\bm E\|_F\le c_E r_k^2$. Also,
\[
\|\bm U\|\;\le\;\|\bm U\|_F
\;\le\;\Big(\sum_{\ell=k-\tau+1}^{k}\|\Delta\bm u_\ell\|^2\Big)^{1/2}
\;\le\;2\alpha\sqrt{\tau}\,r_k .
\]
Consequently, $\|\bm U^\top \bm W\|
\;\le\;\|\bm W\|\,\|\bm U\|
\;\le\;c_U\,r_k$ for some constant $c_U>0$.

Under the sufficient excitation condition in Definition~\ref{defsuff}, there exists \(m>0\) such that \(\sigma_{\min}(\bm U)\ge m\,r_k\). Since $\bm W$ is positive definite, there exists $w_{\min}>0$ with $\lambda_{\min}(\bm W)\ge w_{\min}$, so $\lambda_{\min}({\bm U}^\top \bm W \bm U)\ge w_{\min}\,\sigma_{\min}^2(\bm U)\ge w_{\min} m^2 r_k^2$. Hence
\[
\|({\bm U}^\top \bm W \bm U)^{-1}\|
=\frac{1}{\lambda_{\min}({\bm U}^\top \bm W \bm U)}
\;\le\;\frac{1}{w_{\min} m^2 r_k^2},
\]
where \(\lambda_{\min}(\cdot)\) denotes the smallest eigenvalue and
\(\sigma_{\min}(\cdot)\) denotes the smallest singular value.
Therefore, from \eqref{eq:S_diff_exact},
\begin{equation}\label{eq:S_operator_bound}
\begin{aligned}
&\|\widetilde{\mathbf S}_k-\mathbf S_k\|
\;\le\;\|({\bm U}^\top \bm W \bm U)^{-1}\|\,\|{\bm U}^\top \bm W\|\,\|\bm E\|
\\ &\le\;\frac{1}{w_{\min} m^2 r_k^2}\,\big(c_U r_k\big)\,\big(c_E r_k^2\big)
\;=\;\frac{c_U c_E}{w_{\min} m^2}\,r_k .
\end{aligned}
\end{equation}
Thus, \eqref{eq:S_error_bound} holds with
$C_S:=c_Uc_E/(w_{\min}m^2)$. Consequently, for any $\|\bm u-\bm u_k\|\le r_k$,
\[
\big\|(\widetilde{\mathbf S}_k-\mathbf S_k)(\bm u-\bm u_k)\big\|
\;\le\;C_S r_k\cdot r_k
\;=\;C_Sr_k^2,
\]
which proves \eqref{eq:key_ptwise_goal}. 
% \end{proof}

\subsection{Proof for Lemma~\ref{lem:weak_accept_est}}\label{app:lem_weak_accept_est}

% \begin{proof}
Since $\bar{\bm z}$ is feasible and not stationary, $\bm 0\notin \partial J(\bar{\bm z})$. According to the definition of the generalized differential, $\partial J(\bar{\bm z})$ is a closed convex set; hence, by the separating
hyperplane theorem to $\partial J(\bar{\bm z})$ and the origin, there exist a unit vector $\bm s$ and a constant $\kappa>0$ such that
$\bm{\varphi}^\top \bm s\le -\kappa$ for all $\bm{\varphi}\in\partial J(\bar{\bm z})$. Thus, by the upper semicontinuity of Clarke’s directional derivative, we have
\[
\limsup_{\substack{\bm z\to \bar{\bm z}\\ r\to 0}}
\frac{J(\bm z+r\bm s)-J(\bm z)}{r}
\;\le\; -\kappa,
\]
which implies that there exists 
$\bar\epsilon>0$ and $\bar r>0$ for which
\begin{equation}\label{eq:dir_drop_unit}
J(\bm z+r\bm s)-J(\bm z)\ \le\ -\tfrac{\kappa}{2}\,r
\quad
\forall\,\bm z\in N(\bar{\bm z},\bar\epsilon),\ \forall\,r\in(0,\bar r].
\end{equation}
For any $k$ with $\bm z_k\in N(\bar{\bm z},\bar\epsilon)$ and any $r_k\in(0,\bar r]$.
Consider the trial step $\bm d'_k:=r_k\bm s$, which locates within the trust region of $\bar r$
since $\|\bm d'_k\|=r_k$. Evaluating \eqref{eq:dir_drop_unit} at $\bm z=\bm z_k$ and $r=r_k$ gives
\begin{equation}\label{eq:DeltaJ_trial_k}
\Delta J_k(\bm z_k,\bm d'_k)
:=J(\bm z_k)-J(\bm z_k+\bm d'_k)
\ \ge\ \tfrac{\kappa}{2}\,r_k .
\end{equation}
By Lemma~\ref{lem:Ltilde_minus_L},  for all $\bm d$ with $\|\bm d\|\le r_k$, $J(\bm z_k+\bm d)-\widetilde L_k(\bm d) \ge -C_Mr_k^2.$
Applying this with $\bm d=\bm d'_k$ and rearranging yields
\begin{equation}\label{eq:DeltaL_trial_k}
\begin{aligned}
    \Delta \widetilde L_k(\bm z_k,\bm d'_k)
&= J(\bm z_k)-\widetilde L_k(\bm d'_k) \\
&\ge \Delta J_k(\bm z_k,\bm d'_k)-C_Mr_k^2
\\ &\ge
\tfrac{\kappa}{2}r_k-C_Mr_k^2.
\end{aligned}
\end{equation}
Because $\bm d_k^\star$ minimizes $\widetilde L_k$ over $\{\bm d:\ \|\bm d\|\le r_k\}$,
\begin{equation}\label{eq:predicted_lb_k}
\Delta \widetilde L_k(\bm z_k,\bm d_k^\star)
\ \ge\ \Delta \widetilde L_k(\bm z_k,\bm d'_k)
\ \ge\ \tfrac{\kappa}{2}\,r_k-C_Mr_k^2.
\end{equation}
Again by Lemma~\ref{lem:Ltilde_minus_L}, evaluating \eqref{eq:J_vs_Ltilde} at $\bm d=\bm d_k^\star$ and subtracting from $J(\bm z_k)$ gives
\begin{equation}\label{eq:actual_vs_pred_k}
    \begin{aligned}
    \Delta J_k(\bm z_k,\bm d_k^\star)
&=J(\bm z_k)-J(\bm z_k+\bm d_k^\star) \\
&=
J(\bm z_k)-\widetilde L_k(\bm d_k^\star)
-
\Big(
J(\bm z_k+\bm d_k^\star)
-\widetilde L_k(\bm d_k^\star)
\Big) \\
&\ge
\Delta\widetilde L_k(\bm z_k,\bm d_k^\star)
-
C_Mr_k^2.
\end{aligned}
\end{equation}
Dividing \eqref{eq:actual_vs_pred_k} by \eqref{eq:predicted_lb_k} yields
\[
\rho(\bm z_k,r_k)
=\frac{\Delta J_k(\bm z_k,\bm d_k^\star)}{\Delta \widetilde L_k(\bm z_k,\bm d_k^\star)}
\ \ge
1-
\frac{C_Mr_k^2}
{(\kappa/2)r_k-C_Mr_k^2}.
\]
Dividing the numerator and denominator of the second term by $r_k$ completes the proof.
% \end{proof}

\subsection{Proof for Lemma~\ref{lem:limit points}}\label{app:lem_limit points}
% \begin{proof}
Since the constraint set is compact, the sequence $\{\bm z_k\}$ admits a convergent
subsequence by the Bolzano--Weierstrass theorem. Let $\{\bm z_{k_i}\}$ be any such subsequence with
$\bm z_{k_i}\to\bar{\bm z}$. Define the subsequence Jacobian mismatch level
\[
\bar e := \liminf_{i\to\infty}\|\widetilde{\mathbf S}_{k_i}-\mathbf S_{k_i}\|.
\]
By the definition of the lower limit, there exists a subsequence such that
$\|\widetilde{\mathbf S}_{k_i}-\mathbf S_{k_i}\|\to \bar e$.
We could still denote this subsequence by $\{k_i\}$ for simplicity.

Assume, for contradiction, that $\bar{\bm z}$ lies outside the neighborhood of the stationary set of $J$, namely $\mathrm{dist}\big(\bm 0,\partial J(\bar{\bm z})\big) > \|\bm\lambda\|\,\bar e$. 
% Define $\delta := \mathrm{dist}\big(\bm 0,\partial J(\bar{\bm z})\big)-\|\bm\lambda\|\,\bar e > 0$. 
In particular, $\bm 0\notin\partial J(\bar{\bm z})$.
Since $J$ is locally Lipschitz, $\partial J(\bar{\bm z})$ is a nonempty closed convex set.
Let $\bm s'$ be the projection of $\bm 0$ onto $\partial J(\bar{\bm z})$, so
$\|\bm s'\|=\mathrm{dist}\big(\bm 0,\partial J(\bar{\bm z})\big)$ and $\bm s'\neq \bm 0$.
Set $\bm s:=-\bm s'/\|\bm s'\|$, so $\|\bm s\|=1$.
By the projection variational inequality, for all $\bm\varphi\in\partial J(\bar{\bm z})$, $(\bm\varphi-\bm s')^\top(\bm 0-\bm s')\le 0$, from which we have
\begin{equation}\label{eq:sep_hyp_main}
\bm\varphi^\top \bm s
\le
-\|\bm s'\|
= -\,\mathrm{dist}\big(\bm 0,\partial J(\bar{\bm z})\big) \quad \forall\,\bm\varphi\in\partial J(\bar{\bm z}).
\end{equation}
By outer semicontinuity of the Clarke subdifferential~\cite{RockafellarWets1998VariationalAnalysis}, there exists $\epsilon_1>0$ such that for all $\bm z\in N(\bar{\bm z},\epsilon_1)$ and $ \bm\varphi\in\partial J(\bm z),$
\begin{equation}\label{eq:osc_support}
\bm\varphi^\top \bm s
\le
-\Big(\mathrm{dist}\big(\bm 0,\partial J(\bar{\bm z})\big)-\tfrac{\delta}{4}\Big),
\end{equation}
where we take $\delta := \mathrm{dist}\big(\bm 0,\partial J(\bar{\bm z})\big)-\|\bm\lambda\|\,\bar e > 0$.

Let $\epsilon_1^\prime:=\epsilon_1/2$. For any $\bm z\in N(\bar{\bm z},\epsilon_1^\prime)$ and any $r\in(0,\epsilon_1^\prime]$,
the segment $\{\bm z+t r\bm s:\ t\in[0,1]\}$ is contained in $N(\bar{\bm z},\epsilon_1)$.
By Lebourg's mean value theorem for locally Lipschitz functions\cite{lebourg1979generic}, there exist
$t\in(0,1)$ and $\bm\varphi\in\partial J(\bm z+t r\bm s)$ such that $J(\bm z+r\bm s)-J(\bm z) = r\,\bm\varphi^\top \bm s$. Combining this with \eqref{eq:osc_support} yields
\begin{equation}\label{eq:dir_drop_main}
\begin{aligned}
J(\bm z+r\bm s)-J(\bm z)
&\le
-\Big(\mathrm{dist}\big(\bm 0,\partial J(\bar{\bm z})\big)-\tfrac{\delta}{4}\Big)\,r,
\\
&\forall\,\bm z\in N(\bar{\bm z},\epsilon_1^\prime),\ \forall\,r\in(0,\epsilon_1^\prime].
\end{aligned}
\end{equation}

By Lemma~\ref{lem:weak_accept_est}, there exist $\epsilon_2>0$ and $\bar r>0$ such that
for every $k_i$ with $\bm z_{k_i}\in N(\bar{\bm z},\epsilon_2)$ and every radius
$r_{k_i}\in(0,\bar r]$, any optimal solution $\bm d_{k_i}^\star$ of the trust-region subproblem at
$\bm z_{k_i}$ satisfies
\begin{equation}\label{eq:rho0-k}
\frac{\Delta J_{k_i}(\bm z_{k_i},\bm d_{k_i}^\star)}{\Delta \widetilde L_{k_i}(\bm z_{k_i},\bm d_{k_i}^\star)}
\ge
\rho_0>0,
\end{equation}
where $\rho_0$ is a positive constant. Without loss of generality, assume each $\bm z_{k_i}$ in the subsequence
locates within $N(\bar{\bm z},\bar\epsilon)$ with $\bar\epsilon=\min\{\epsilon_1^\prime,\epsilon_2\}$, and the trust-region radius
$r_k$ satisfies $0<\underline r \le r_{k} \le \bar r$ for all $k$.

Next, fix a constant step length $\widehat r>0$ satisfying
\[
\widehat r \le \min\{\underline r,\,\epsilon_1^\prime\}.
\]
Then $\|\widehat r\,\bm s\|=\widehat r\le r_{k_i}$, so the trial step $\bm d'_{k_i}:=\widehat r\,\bm s$ is feasible for every trust region along the subsequence. Since $\|\widetilde{\mathbf S}_{k_i}-\mathbf S_{k_i}\|\to\bar e$, there exists $i_0$ such that, for all $i\ge i_0$,
\begin{equation}\label{eq:mismatch_close}
\|\widetilde{\mathbf S}_{k_i}-\mathbf S_{k_i}\|
\le
\bar e+\frac{\delta}{4\|\bm\lambda\|}.
\end{equation}
Define
\begin{equation}\label{eq:theta-def}
\begin{aligned}
\theta
:&= \liminf_{i\to\infty}\Delta \widetilde L_{k_i}(\bm z_{k_i},\bm d_{k_i}^\star)
\\
&= \liminf_{i\to\infty}\Big(J(\bm z_{k_i})-\widetilde L_{k_i}(\bm d_{k_i}^\star)\Big).
\end{aligned}
\end{equation}
By Lemma~\ref{lem:pred_nonneg}, $\Delta \widetilde L_{k_i}(\bm z_{k_i},\bm d_{k_i}^\star)\ge 0$ for all $i$, hence $\theta\ge 0$.
We now show that $\theta>0$. Fix $i\ge i_0$. As $\bm z_{k_i}\in N(\bar{\bm z},\bar\epsilon)$,
evaluating \eqref{eq:dir_drop_main} at $\bm z=\bm z_{k_i}$ and $r=\widehat r$ gives
\begin{equation}\label{eq:DeltaJ_trial_main}
\begin{aligned}
\Delta J_{k_i}(\bm z_{k_i},\bm d'_{k_i})
&= J(\bm z_{k_i})-J(\bm z_{k_i}+\bm d'_{k_i})
\\
&\ge
\Big(\mathrm{dist}\big(\bm 0,\partial J(\bar{\bm z})\big)-\tfrac{\delta}{4}\Big)\,\widehat r.
\end{aligned}
\end{equation}
To relate $J(\bm z_{k_i}+\bm d'_{k_i})$ and $\widetilde L_{k_i}(\bm d'_{k_i})$, write
\begin{equation}\label{eq:decomp_err}
\begin{aligned}
J(\bm z_{k_i}+\bm d'_{k_i})-\widetilde L_{k_i}(\bm d'_{k_i})
&=
\Big(J(\bm z_{k_i}+\bm d'_{k_i})-L_{k_i}(\bm d'_{k_i})\Big)
\\
&\ +
\Big(L_{k_i}(\bm d'_{k_i})-\widetilde L_{k_i}(\bm d'_{k_i})\Big).
\end{aligned}
\end{equation}
By Lemma~\ref{lem:taylor}, there exists $r_0>0$ such that for all $\|\bm d\|\le r_0$,  $\Big|J(\bm z_{k_i}+\bm d)-L_{k_i}(\bm d)\Big|
\le
\tfrac{\delta}{4}\,\|\bm d\|$. Take $\widehat r\le r_0$, so this applies to $\bm d=\bm d'_{k_i}$ and yields
\begin{equation*}
    \Big|J(\bm z_{k_i}+\bm d'_{k_i})-L_{k_i}(\bm d'_{k_i})\Big|
\le
\tfrac{\delta}{4}\,\widehat r.
\end{equation*}
% \[
% \Big|J(\bm z_{k_i}+\bm d'_{k_i})-L_{k_i}(\bm d'_{k_i})\Big|
% \le
% \tfrac{\delta}{4}\,\widehat r.
% \]
Since $\phi_{\mathbf M}(\bm d)=\bm\lambda^\top|\bm d_y-\mathbf M\bm d_u|$ is Lipschitz in $\mathbf M$,
\begin{equation}\label{eq:model_bd_main}
\begin{aligned}
\Big|L_{k_i}(\bm d'_{k_i})-\widetilde L_{k_i}(\bm d'_{k_i})\Big|
&=
\Big|\phi_{\mathbf S_{k_i}}(\bm d'_{k_i})-\phi_{\widetilde{\mathbf S}_{k_i}}(\bm d'_{k_i})\Big|
\\
&\le
\|\bm\lambda\|\,\|\widetilde{\mathbf S}_{k_i}-\mathbf S_{k_i}\|\,\|\bm d'_{k_i}\|
\\
&\le
\Big(\|\bm\lambda\|\,\bar e+\tfrac{\delta}{4}\Big)\,\widehat r,
\end{aligned}
\end{equation}
where the last line uses \eqref{eq:mismatch_close}. Now, use $\Delta \widetilde L_{k_i}(\bm z_{k_i},\bm d'_{k_i})
=
\Delta J_{k_i}(\bm z_{k_i},\bm d'_{k_i}) +
\Big(J(\bm z_{k_i}+\bm d'_{k_i})-\widetilde L_{k_i}(\bm d'_{k_i})\Big)$,
together with \eqref{eq:DeltaJ_trial_main}, \eqref{eq:decomp_err},  and the bounds above, to obtain
\begin{equation}\label{eq:DeltaL_trial_main}
\begin{aligned}
\Delta \widetilde L_{k_i}(\bm z_{k_i},\bm d'_{k_i})
&\ge
\Big(\mathrm{dist}\big(\bm 0,\partial J(\bar{\bm z})\big)-\tfrac{\delta}{4}\Big)\widehat r
\\
&\quad
-\tfrac{\delta}{4}\widehat r
-\Big(\|\bm\lambda\|\,\bar e+\tfrac{\delta}{4}\Big)\widehat r
\\
&=
\Big(\mathrm{dist}\big(\bm 0,\partial J(\bar{\bm z})\big)-\|\bm\lambda\|\,\bar e-\tfrac{3\delta}{4}\Big)\widehat r
\\
&=
\tfrac{\delta}{4}\,\widehat r
>0.
\end{aligned}
\end{equation}
Since $\bm d_{k_i}^\star$ minimizes $\widetilde L_{k_i}$ over $\{\bm d:\ \|\bm d\|\le r_{k_i}\}$, we have $\Delta \widetilde L_{k_i}(\bm z_{k_i},\bm d_{k_i}^\star)
\ge
\Delta \widetilde L_{k_i}(\bm z_{k_i},\bm d'_{k_i})
\ge
\tfrac{\delta}{4}\,\widehat r$. Taking the lower limit yields $\theta\ge \tfrac{\delta}{4}\widehat r>0$, hence $\theta>0$.

By the definition of $\theta$ in \eqref{eq:theta-def} and $\theta>0$, there exists $i_2$ such that for all $i\ge i_2$, $\Delta \widetilde L_{k_i}(\bm z_{k_i},\bm d_{k_i}^\star)
\ge
\theta/2
>0$.
Since $\bm z_{k_i}\in N(\bar{\bm z},\bar\epsilon)$ and $r_{k_i}\in(0,\bar r]$, the lower bound \eqref{eq:rho0-k} applies; therefore,
\begin{equation}\label{eq:actual-dec}
\begin{aligned}
J(\bm z_{k_i})-J(\bm z_{k_i+1})
&=\Delta J_{k_i}(\bm z_{k_i},\bm d_{k_i}^\star)
\\
&\ge
\rho_0\,\Delta \widetilde L_{k_i}(\bm z_{k_i},\bm d_{k_i}^\star)
\\
&\ge
\rho_0\,\theta/2
>0,
\end{aligned}
\end{equation}
where the first equality uses the update $\bm z_{k_i+1}=\bm z_{k_i}+\bm d_{k_i}^\star$.
From \eqref{eq:actual-dec}, we obtain $J(\bm z_{k+1})\le J(\bm z_k)$ for all $k$.
As $k_i+1\le k_{i+1}$, we have $J(\bm z_{k_i+1})\ge J(\bm z_{k_{i+1}})$.
Hence, for any $N\ge i_2$, $\sum_{i=i_2}^{N}\big(J(\bm z_{k_i})-J(\bm z_{k_i+1})\big)
\le
J(\bm z_{k_{i_2}})-J(\bm z_{k_{N+1}})$. Letting $N\to\infty$ and using $\bm z_{k_i}\to\bar{\bm z}$ yields
\[
\sum_{i=i_2}^{\infty}\big(J(\bm z_{k_i})-J(\bm z_{k_i+1})\big)
\le
J(\bm z_{k_{i_2}})-J(\bar{\bm z})
<\infty.
\]
Therefore the series on the left converges, and necessarily $J(\bm z_{k_i})-J(\bm z_{k_i+1})\to 0$, which contradicts the uniform lower bound \eqref{eq:actual-dec}.
Consequently, the contradiction assumption is false. Thus,
\[
\mathrm{dist}\big(\bm 0,\partial J(\bar{\bm z})\big)
\le
\|\bm\lambda\|\,\liminf_{i\to\infty}\|\widetilde{\mathbf S}_{k_i}-\mathbf S_{k_i}\|.
\]
Equivalently, $\bar{\bm z}$ is $\|\bm\lambda\|\liminf_{i\to\infty}\|\widetilde{\mathbf S}_{k_i}-\mathbf S_{k_i}\|$-close to the stationary set $\mathcal T$ of the penalty problem \eqref{eq:penalty_cost}.
% \end{proof}

% ==================== Appendix ====================
\subsection{Time-Varying Case Study with 10 Measured and Controlled Buses}\label{app:partial}

The following figures provide the voltage and full-network cost trajectories
for the more restrictive case in which only 10 buses are measured and
controlled. The proposed data-driven successive linearization continues to maintain a relatively low full network cost, while a small number of voltage violations begin to appear. The deviations are more pronounced at unmeasured and uncontrolled buses, illustrating the degradation in regulation performance under more limited system information and control capability. However, voltages of data-driven successvie linearization still generally remain within 0.95--1.05~p.u. limits.

\begin{figure}[h]
  \centering
  \includegraphics[width=0.95\linewidth]{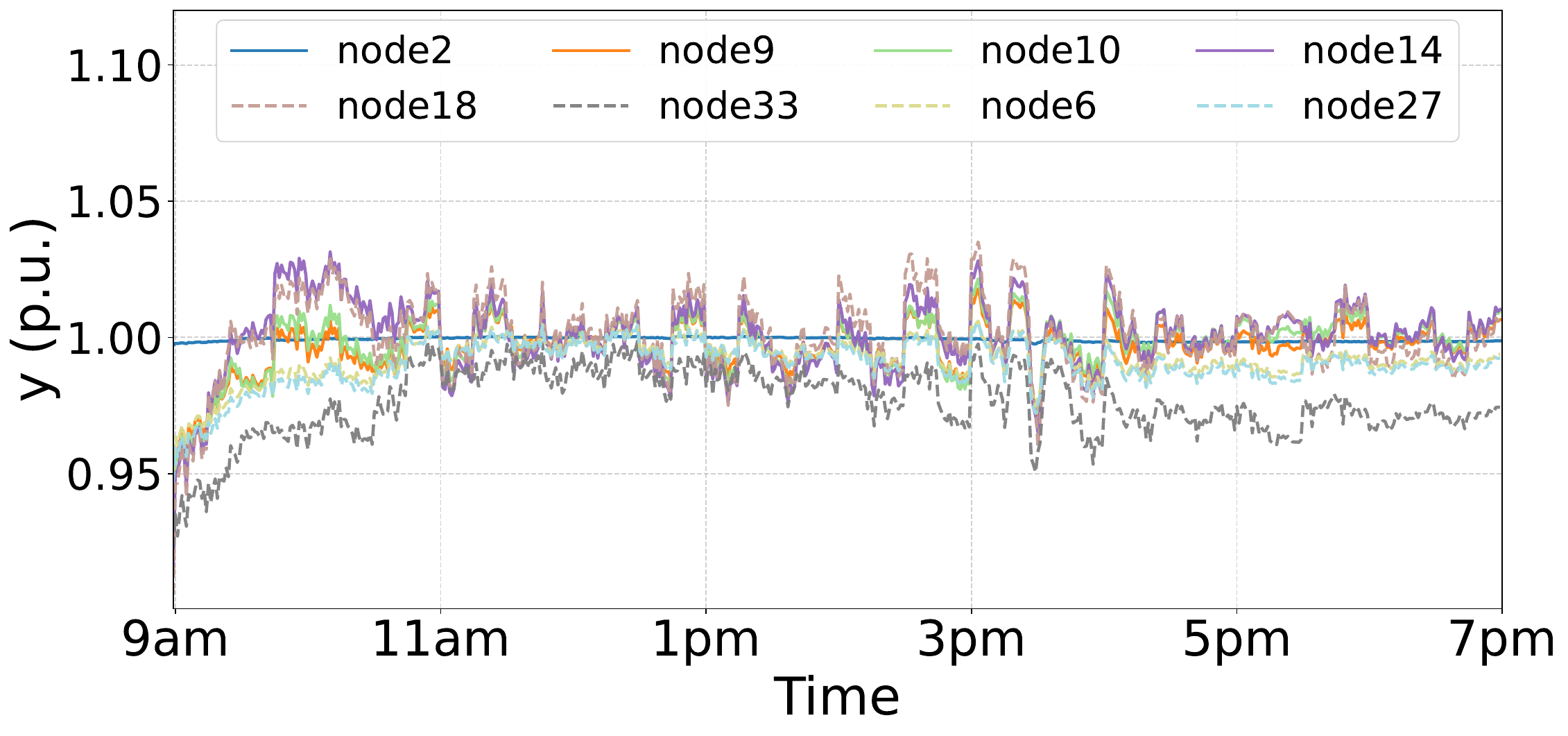}
  \vspace{-0.3cm}
  \caption{Voltage trajectories produced by data-driven successive linearization under time-varying loads when 10 buses are measured and controlled. Solid lines represent selected measured-and-controlled buses, whereas dashed lines represent selected unmeasured-and-uncontrolled buses.}
  \vspace{-0.4cm}
  \label{fig:partial10_ddsl_voltage}
\end{figure}

\begin{figure}[h]
  \centering
  \includegraphics[width=0.95\linewidth]{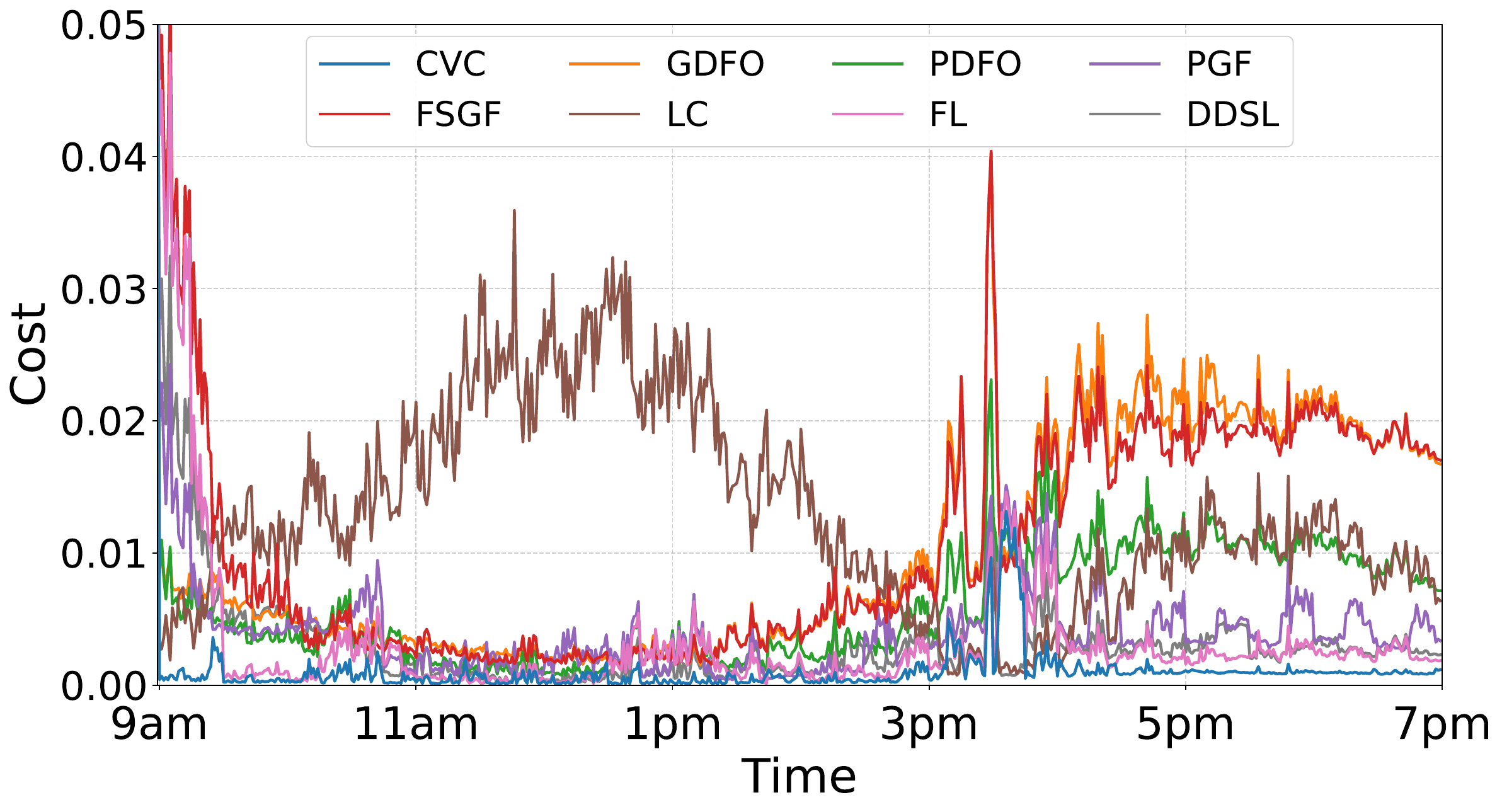}
  \vspace{-0.3cm}
  \caption{Cost trajectories of the eight controllers under time-varying loads when 10 buses are measured and controlled. Convex relaxation is included as a full-information benchmark, while the other controllers use measurements only from the designated subset. All methods apply reactive power actions only at the same 10 controlled buses every 15 minutes.}
  \vspace{-0.4cm}
  \label{fig:partial10_cost}
\end{figure}

\subsection{Measurement Noise Analysis}
\label{app:measurement_noise}

This section analyzes the effect of measurement noise on Jacobian estimation. We consider measured voltages $\widehat{\bm y}_\ell=\bm y_\ell+\bm e_\ell$, where $\bm y_\ell$ is the true voltage and $\bm e_\ell$ is the measurement noise. The effect of noise on Jacobian estimation is
summarized in the following proposition.

\noindent\textbf{Proposition.}
Suppose the assumptions of Lemma~1 hold. Assume that the voltage measurement noise $\bm e_\ell$ is bounded by $r_m\ge0$, namely  $\|\bm e_\ell\|\le r_m$ for all $\ell$. Then, the Jacobian estimate under measurement noise satisfies
\begin{equation}\label{eq:app_noise_jacobian_bound}
\|\widetilde{\mathbf S}_k-\mathbf S_k\|
\le C_Sr_k+C_m\frac{r_m}{r_k},
\end{equation}
where $r_k>0$ is the trust region radius and $C_S,C_m>0$ are constants.

The first term is associated with the local approximation error, while the second is caused by measurement noise. Thus, for a fixed
trust region radius, increasing the noise level $r_m$ increases the Jacobian
estimation error. The bound also shows that choosing $r_k$ too small can
amplify the effect of measurement noise on Jacobian estimation. The detailed proof of this proposition is
provided below.
\begin{proof}
Let $\widehat{\bm y}_\ell$ denote the measured voltage vector at iteration
$\ell$. The hat symbol distinguishes a measured quantity from its true,
noise-free value $\bm y_\ell$. We assume bounded additive measurement noise:
$\widehat{\bm y}_\ell=\bm y_\ell+\bm e_\ell$,
$\|\bm e_\ell\|\le r_m$, where $\bm e_\ell\in\mathbb R^n$ is the voltage
measurement error and $r_m\ge0$ is a uniform bound on its Euclidean norm.
The control inputs are the reactive power setpoints issued by the
controller and are assumed to be accurately recorded. The measured voltage
increment is
$\Delta\widehat{\bm y}_\ell
:=\widehat{\bm y}_\ell-\widehat{\bm y}_{\ell-1}$.
Substituting the measurement model gives
$\Delta\widehat{\bm y}_\ell
=(\bm y_\ell+\bm e_\ell)-(\bm y_{\ell-1}+\bm e_{\ell-1})
=\Delta\bm y_\ell+\bm\nu_\ell$,
where $\bm\nu_\ell:=\bm e_\ell-\bm e_{\ell-1}$ is the measurement error
affecting the voltage increment. Since each individual voltage measurement
error is bounded by $r_m$,
$\|\bm\nu_\ell\|
\le\|\bm e_\ell\|+\|\bm e_{\ell-1}\|
\le2r_m$.

Define the control-increment matrix
$\bm U_k\in\mathbb R^{\tau\times n_c}$ by stacking the rows
$\Delta\bm u_\ell^\top$ for $\ell=k-\tau+1,\ldots,k$, and we assume
$\|\Delta\bm u_\ell\|\le2\alpha r_k$ as in Lemma~1.
Define the true voltage-increment matrix
$\bm Y_k\in\mathbb R^{\tau\times n_o}$ by stacking the corresponding rows
$\Delta\bm y_\ell^\top$. Similarly, define the measured voltage-increment matrix and  the measurement-noise matrix
\[
\widehat{\bm Y}_k:=
\begin{bmatrix}
\Delta\widehat{\bm y}_{k-\tau+1}^{\top}\\[-0.25em]
\vdots\\[-0.25em]
\Delta\widehat{\bm y}_k^{\top}
\end{bmatrix}
\in\mathbb R^{\tau\times n_o}, \bm N_k:=
\begin{bmatrix}
\bm\nu_{k-\tau+1}^{\top}\\[-0.25em]
\vdots\\[-0.25em]
\bm\nu_k^{\top}
\end{bmatrix}
\in\mathbb R^{\tau\times n_o}.
\]

By construction, $\widehat{\bm Y}_k=\bm Y_k+\bm N_k$. Moreover, since
each of the $\tau$ rows of $\bm N_k$ has norm at most $2r_m$,
$\|\bm N_k\|
\le\|\bm N_k\|_F
=\bigl(\sum_{\ell=k-\tau+1}^{k}\|\bm\nu_\ell\|^2\bigr)^{1/2}
\le\bigl(\tau(2r_m)^2\bigr)^{1/2}
=2\sqrt{\tau}\,r_m$.

We next express the true voltage-increment matrix in terms of the current
Jacobian $\mathbf S_k$. For each $\ell$,
\begin{align*}
\Delta\bm y_\ell
&=\mathbf S_\ell\Delta\bm u_\ell+\bm\gamma_\ell\\
&=\mathbf S_k\Delta\bm u_\ell
+(\mathbf S_\ell-\mathbf S_k)\Delta\bm u_\ell
+\bm\gamma_\ell.
\end{align*}
After transposing this identity and stacking all $\tau$ rows, we obtain
$\bm Y_k=\bm U_k\mathbf S_k^\top+\bm E_k$, where
\[
\bm E_k:=
\begin{bmatrix}
\Delta\bm u_{k-\tau+1}^{\top}
(\mathbf S_{k-\tau+1}^{\top}-\mathbf S_k^\top)
+\bm\gamma_{k-\tau+1}^{\top}\\[-0.1em]
\vdots\\[-0.1em]
\Delta\bm u_k^{\top}
(\mathbf S_k^\top-\mathbf S_k^\top)
+\bm\gamma_k^{\top}
\end{bmatrix}
\in\mathbb R^{\tau\times n_o}.
\]

Combining $\widehat{\bm Y}_k=\bm Y_k+\bm N_k$ with
$\bm Y_k=\bm U_k\mathbf S_k^\top+\bm E_k$ gives the noisy data
decomposition
\[
\widehat{\bm Y}_k
=\bm U_k\mathbf S_k^\top+\bm E_k+\bm N_k.
\]

We now bound the local linearization-error matrix $\bm E_k$. The local
Lipschitz continuity of the Jacobian gives
$\|\mathbf S_\ell-\mathbf S_k\|
\le L_S\|\bm u_\ell-\bm u_k\|
\le L_S\alpha r_k$.
Together with $\|\Delta\bm u_\ell\|\le2\alpha r_k$, this implies
$\bigl\|\Delta\bm u_\ell^\top
(\mathbf S_\ell^\top-\mathbf S_k^\top)\bigr\|
\le\|\Delta\bm u_\ell\|\|\mathbf S_\ell-\mathbf S_k\|
\le(2\alpha r_k)(L_S\alpha r_k)
=2L_S\alpha^2r_k^2$.
Combining this bound with
$\|\bm\gamma_\ell\|\le4\beta\alpha^2r_k^2$ shows that every row of
$\bm E_k$ has norm bounded by
$2L_S\alpha^2r_k^2+4\beta\alpha^2r_k^2
=(2L_S+4\beta)\alpha^2r_k^2$.
Define $c_{E,0}:=(2L_S+4\beta)\alpha^2$ and
$c_E:=\sqrt{\tau}\,c_{E,0}$, then
$\|\bm E_k\|
\le\|\bm E_k\|_F
\le\sqrt{\tau}\,c_{E,0}r_k^2
=c_Er_k^2$.

Let
$\bm W_k=\operatorname{diag}(w_{k-\tau+1},\ldots,w_k)
\in\mathbb R^{\tau\times\tau}$
be the diagonal weight matrix used in weighted least squares. Assume that
the weights are strictly positive and uniformly bounded:
$0<w_{\min}\le\lambda_{\min}(\bm W_k)$,
$\|\bm W_k\|\le w_{\max}$ for some constants $w_{\min},w_{\max}>0$.
The Jacobian estimate under measurement noise is now
\[
\widetilde{\mathbf S}_k^\top
=(\bm U_k^\top\bm W_k\bm U_k)^{-1}
\bm U_k^\top\bm W_k\widehat{\bm Y}_k.
\]

Substituting
$\widehat{\bm Y}_k=\bm U_k\mathbf S_k^\top+\bm E_k+\bm N_k$
into the estimator gives
\begin{align*}
\widetilde{\mathbf S}_k^\top
&=(\bm U_k^\top\bm W_k\bm U_k)^{-1}
\bm U_k^\top\bm W_k\bigl(\bm U_k\mathbf S_k^\top+\bm E_k+\bm N_k\bigr)\\
&=\mathbf S_k^\top
+(\bm U_k^\top\bm W_k\bm U_k)^{-1}
\bm U_k^\top\bm W_k
(\bm E_k+\bm N_k).
\end{align*}
Therefore,
\[
\widetilde{\mathbf S}_k^\top-\mathbf S_k^\top
=(\bm U_k^\top\bm W_k\bm U_k)^{-1}
\bm U_k^\top\bm W_k
(\bm E_k+\bm N_k).
\]

We next bound each factor on the right-hand side. First,
$\|\bm U_k\|
\le\|\bm U_k\|_F
=\bigl(\sum_{\ell=k-\tau+1}^{k}
\|\Delta\bm u_\ell\|^2\bigr)^{1/2}
\le\bigl(\tau(2\alpha r_k)^2\bigr)^{1/2}
=2\alpha\sqrt{\tau}\,r_k$.
It follows that
$\|\bm U_k^\top\bm W_k\|
\le\|\bm U_k\|\,\|\bm W_k\|
\le2\alpha\sqrt{\tau}\,w_{\max}r_k$.
Define $c_U:=2\alpha\sqrt{\tau}\,w_{\max}$. Then
$\|\bm U_k^\top\bm W_k\|\le c_Ur_k$.

As the sufficient excitation condition holds,
$\sigma_{\min}(\bm U_k)\ge mr_k$ for some constant $m>0$, where
$\sigma_{\min}(\bm U_k)$ denotes the smallest singular value of
$\bm U_k$.

Since $\bm W_k$ is positive definite,
$\lambda_{\min}(\bm U_k^\top\bm W_k\bm U_k)
\ge\lambda_{\min}(\bm W_k)\sigma_{\min}^2(\bm U_k)
\ge w_{\min}m^2r_k^2$. Consequently,
\begin{align*}
\bigl\|(\bm U_k^\top\bm W_k\bm U_k)^{-1}\bigr\|
=\frac{1}
{\lambda_{\min}(\bm U_k^\top\bm W_k\bm U_k)}\le\frac{1}{w_{\min}m^2r_k^2}.
\end{align*}

Combining all of the preceding bounds yields
\begin{align*}
\|\widetilde{\mathbf S}_k-\mathbf S_k\|
&=\|\widetilde{\mathbf S}_k^\top-\mathbf S_k^\top\|\\
&\le
\bigl\|(\bm U_k^\top\bm W_k\bm U_k)^{-1}\bigr\|
\|\bm U_k^\top\bm W_k\|\bigl(\|\bm E_k\|+\|\bm N_k\|\bigr)\\
&\le
\frac{1}{w_{\min}m^2r_k^2}
(c_Ur_k)
\bigl(c_Er_k^2+2\sqrt{\tau}\,r_m\bigr)\\
&=
\frac{c_Uc_E}{w_{\min}m^2}r_k
+\frac{2\sqrt{\tau}\,c_U}{w_{\min}m^2}
\frac{r_m}{r_k}.
\end{align*}
Define $C_S:=c_Uc_E/(w_{\min}m^2)$ and
$C_m:=2\sqrt{\tau}\,c_U/(w_{\min}m^2)$.
We therefore obtain the Jacobian estimation error bound
\[
\|\widetilde{\mathbf S}_k-\mathbf S_k\|
\le C_Sr_k+C_m\frac{r_m}{r_k},
\]
\end{proof}

\subsection{Discussion on Perturbation Injection}
\label{app:perturbation_excitation}

We characterize how injecting perturbations into control action affects Jacobian estimation. In summary, if the
perturbation scale $\widetilde r$ is too small, the control
increments may become small, which can reduce $\sigma_{\min}(\bm U)$ and  lead to ill-conditioned sensitivity estimation. Conversely, increasing
$\widetilde r$ can avoid ill-conditioned sensitivity estimation, but it also enlarges the range of data used for estimation and may also increase Jacobian estimation error. In practice, variations in load and distributed generation in power system can naturally serve as perturbations around the operating points, which provides the excitation for Jacobian estimation.

\noindent\textbf{Proposition.}
Suppose $\|\bm u_\ell-\bm u_k\|\le\alpha\widetilde r$ for
$\ell=k-\tau,\ldots,k$, where $\alpha,\widetilde r>0$, and $\bm U$
has full column rank. Then the Jacobian estimation error satisfies
\begin{equation}\label{eq:app_perturbation_bound}
\|\widetilde{\mathbf S}_k-\mathbf S_k\|
\le C\frac{\widetilde r^3}{\sigma_{\min}^2(\bm U)},
\end{equation}
where $C>0$ is a constant and $\sigma_{\min}(\bm U)$ is the smallest
singular value of $\bm U$.

The bound reflects the tradeoff between excitation and locality.
The perturbation scale $\widetilde r$ determines the neighborhood covered
by the estimation data, while $\sigma_{\min}(\bm U)$ measures the
excitation provided by the control increments. When $\widetilde r$ is chosen very small, the resulting control increments
also become small because $\|\Delta\bm u_\ell\|\le 2\alpha\widetilde r$. Consequently, $\sigma_{\min}(\bm U)$ may decrease, and the weighted
least-squares problem may become poorly conditioned, so the Jacobian estimation error increases. On the other hand, increasing the magnitude of sufficiently diverse
perturbations can increase $\sigma_{\min}(\bm U)$ and thereby improve the
conditioning of the weighted least-squares problem. However, it also
increases $\widetilde r$, so that the identification data are collected over
a larger neighborhood in which the Jacobian varies more significantly, which
increases the local linearization error and may also increase the Jacobian estimation error. Thus, the perturbation magnitude is supposed to be neither excessively large nor excessively small. The full proof is given below.

\begin{proof}
From $\|\bm u_\ell-\bm u_k\|\le\alpha\widetilde r$,
$\ell=k-\tau+1,\ldots,k$, we have
$\|\Delta\bm u_\ell\|
\le\|\bm u_\ell-\bm u_k\|
+\|\bm u_{\ell-1}-\bm u_k\|
\le2\alpha\widetilde r$.
Following the notation of Lemma~1, we write
$\Delta\bm y_\ell
=\mathbf S_\ell\Delta\bm u_\ell+\bm\gamma_\ell$,
$\|\bm\gamma_\ell\|\le\beta\|\Delta\bm u_\ell\|^2$
according to first-order Taylor expansion, and consequently
\[
\bm Y=\bm U\mathbf S_k^\top+\bm E,
\]
where $\bm U,\bm Y$ stack the rows
$\Delta\bm u_\ell^\top$ and $\Delta\bm y_\ell^\top$, respectively,
for $\ell=k-\tau+1,\ldots,k$, and
\begin{equation*}
\bm E:=
\begin{bmatrix}
\Delta\bm u_{k-\tau+1}^\top
(\mathbf S_{k-\tau+1}^\top-\mathbf S_k^\top)
+\bm\gamma_{k-\tau+1}^\top\\[-0.25em]
\vdots\\[-0.25em]
\Delta\bm u_k^\top
(\mathbf S_k^\top-\mathbf S_k^\top)
+\bm\gamma_k^\top
\end{bmatrix}.
\end{equation*}

The weighted least-squares estimator satisfies
\[
\widetilde{\mathbf S}_k^\top-\mathbf S_k^\top
=(\bm U^\top\bm W\bm U)^{-1}
\bm U^\top\bm W\bm E.
\]

Since the Jacobian is locally Lipschitz,
$\|\mathbf S_\ell-\mathbf S_k\|
\le L_S\|\bm u_\ell-\bm u_k\|
\le L_S\alpha\widetilde r$.
Together with $\|\Delta\bm u_\ell\|\le2\alpha\widetilde r$,
this gives
$\bigl\|\Delta\bm u_\ell^\top
(\mathbf S_\ell^\top-\mathbf S_k^\top)\bigr\|
\le2L_S\alpha^2\widetilde r^2$,
while $\|\bm\gamma_\ell\|\le4\beta\alpha^2\widetilde r^2$.
Hence, for some constants $c_E,c_U>0$,
$\|\bm E\|\le c_E\widetilde r^2$ and
$\|\bm U^\top\bm W\|\le c_U\widetilde r$.

Since $\bm W$ is positive definite with
$\lambda_{\min}(\bm W)\ge w_{\min}>0$, we first have
$\lambda_{\min}(\bm U^\top\bm W\bm U)
\ge w_{\min}\sigma_{\min}^2(\bm U)$,
and hence
\[
\bigl\|(\bm U^\top\bm W\bm U)^{-1}\bigr\|
\le\frac{1}{w_{\min}\sigma_{\min}^2(\bm U)}.
\]
Combining this inequality with
$\|\bm U^\top\bm W\|\le c_U\widetilde r$ and
$\|\bm E\|\le c_E\widetilde r^2$ gives
\begin{align*}
\|\widetilde{\mathbf S}_k-\mathbf S_k\|
&\le
\bigl\|(\bm U^\top\bm W\bm U)^{-1}\bigr\|
\|\bm U^\top\bm W\|\,\|\bm E\|\\
&\le
\frac{c_Uc_E}{w_{\min}}
\frac{\widetilde r^3}{\sigma_{\min}^2(\bm U)}.
\end{align*}
\end{proof}